\documentclass[12pt,a4paper, total={7in, 8in}]{article}
\usepackage{listings}
\usepackage{blindtext}
\usepackage{amssymb}
\usepackage{amsmath}
\usepackage{graphicx} 
\usepackage{csquotes}
\usepackage[usenames,dvipsnames]{color}
\usepackage{amsthm}
\usepackage{bbm}
\usepackage{float}
\newcommand*\diff{\mathop{}\!\mathrm{d}}

\newtheorem{theorem}{Theorem}
\newtheorem{lemma}{Lemma}
\newtheorem{proposition}{Proposition}

\newtheorem{remark}{Remark}

\newtheorem{definition}{Definition}
\newtheorem{assumption}{Assumption}

\usepackage{enumitem}
\newcommand{\subscript}[2]{$#1 _ #2$}

\makeatletter
\newcommand{\setword}[2]{%
  \phantomsection
  #1\def\@currentlabel{\unexpanded{#1}}\label{#2}%
}
\makeatother

\def\<#1>{\mathinner{\langle#1\rangle}}
\providecommand{\keywords}[1]
{
  \small	
  \textbf{\textit{Keywords---}} #1
} 

\usepackage[margin=0.9in,top=1in,bottom=1in]{geometry}

\title{Numerical analysis of a particle system for the calibrated Heston-type local stochastic volatility model}

\author{Christoph Reisinger\thanks{
Mathematical Institute, University of Oxford, Oxford OX2 6GG, UK
 ({\tt christoph.reisinger@maths.ox.ac.uk, 
maria.tsianni@maths.ox.ac.uk})}
\and
Maria Olympia Tsianni\footnotemark[1]
}
\date{\today}

\begin{document}

\maketitle
\begin{abstract}
We analyse a Monte Carlo particle method for the simulation of the calibrated Heston-type local stochastic volatility (H-LSV) model.
The common application of a kernel estimator for a conditional expectation in the calibration condition 
results in a McKean--Vlasov (MV) stochastic differential equation (SDE) with non-standard coefficients. The primary challenges lie in certain mean-field terms in the drift and diffusion coefficients and the $1/2$-H\"{o}lder regularity of the diffusion coefficient. We establish the well-posedness of this equation for a fixed but arbitrarily small bandwidth of the kernel estimator. Moreover, we
prove a strong propagation of chaos result, ensuring convergence of the particle system under a condition on the Feller ratio and up to a critical time. For the numerical simulation, we employ an Euler--Maruyama scheme for the log-spot process and a full truncation Euler scheme for the CIR volatility process. Under certain conditions on the inputs and the Feller ratio, we prove strong convergence of the Euler--Maruyama scheme with rate $1/2$ in time, up to a logarithmic factor. Numerical experiments illustrate the convergence of the discretisation scheme and validate the propagation of chaos in practice.
\end{abstract}
\section{Introduction}

\hspace*{1.5em} Empirical, theoretical, and computational analysis of volatility models for derivative pricing is a classic topic that has given rise
to comprehensive, deep, and elegant mathematical results over the past decades. Yet, there are still intriguing open questions at the core of the well-posedness
and simulation of some of the most prevalent models, as is the case for the calibrated dynamics of the local stochastic volatility (LSV) model we consider here.\bigbreak


The class of LSV models combines features from both stochastic volatility (SV) and local volatility (LV) models, thereby enjoying the principal advantages of each. Stochastic volatility models (e.g., the Heston model~\cite{SHeston}) capture crucial market phenomena, such as clustered volatility and a negative correlation between volatility and asset price. However, as a parametric family with finitely many parameters, they cannot capture the full range of implied volatility smiles. By contrast, local volatility models (e.g., Dupire~\cite{BDup}) achieve perfect calibration to vanilla options but may fail to capture more complex market dynamics. To balance these strengths and limitations, the LSV model, first introduced in~\cite{JexHenWan}, to the best of our knowledge, combines both volatility components within a single framework, offering a powerful tool for practical applications. Empirical studies and industry practice suggest that this hybrid approach often outperforms either pure LV or SV models in the calibration of European options, as well as in pricing and risk management (see, e.g.,~\cite{lipton, renetal, Tianetal, Stoepetal, GuyHen2}).
\newline 



 In this paper, we focus specifically on the Heston-type local-stochastic volatility (H-LSV) model, in which the squared volatility process has Cox--Ingersoll--Ross (CIR) dynamics (see \cite{CIR}). This model is popular in the industry due to the useful properties of the CIR process, such as mean-reversion and non-negativity, and because the Heston model is analytically tractable, allowing for rapid calibration of the Heston parameters. \bigbreak
 
For a given time horizon $[0, T]$, we consider a complete filtered probability space $(\Omega, \mathcal{F}, \mathbb{F}, \mathbb{Q})$ with natural filtration $\mathbb{F} = (\mathcal{F}_t)_{t \in [0, T]}$, which supports an $\mathcal{F}_t$--adapted two-dimensional standard Brownian motion $(W^s, W^v)$. The H-LSV model has the risk-neutral dynamics:
\begin{equation}\label{HestonLVM}
    \begin{split}
    &\diff S_t = \sqrt{V_t}\,S_t\, \sigma(t,S_t) \diff W^s_t, \, S_0>0,\\
    &\diff V_t = k(\theta-V_t)\diff t + \xi\sqrt{V_t}\diff W^v_t,\, V_0>0,\\
    &\diff W^s_t\diff W^v_t = \rho\diff t, \, \rho \in (-1,1),
    \end{split}
\end{equation}
where $S_t$ is the value at time $t$ of the one-dimensional process $S = (S_t)_{t\in [0,T]}$ describing the spot price of the underlying asset, and $V_t$ is the value at time $t$ of the squared volatility process $V = (V_t)_{t\in[0,T]}$. The volatility process has parameters $k, \theta, \xi > 0$ describing its mean-reversion rate, long-term mean, and volatility, respectively.\bigbreak

To calibrate this model to market prices, it is common to follow two steps. Initially, using a set of observed call option prices from the market, one calibrates a pure Heston process to determine parameters under which the model best matches the market prices. Here, a semi-closed form solution is available for the price as a function of the parameters, and optimal parameters can be found by solving a low-dimensional constrained optimisation problem numerically. Second, at each step $t$ of a time-discretisation, one needs to calibrate the so-called \emph{leverage function} $\sigma(t,\cdot)$. A consistency condition for the exact calibration of SV models to market prices is formulated by Dupire in \cite{BDupire} using Gy\"{o}ngy's result \cite{Gyo}. The specification for general LSV models, as outlined in \cite{JexHenWan}, and tailored for the H-LSV model is given by, for all $s>0$ and $t$ up to the largest traded maturity:

\begin{equation}\label{Dupireiff}
\sigma^{2}(t, s)=
\frac{\sigma_{\text{Dup}}^{2}(t,s)}{\mathbb{E}^{\mathbb{Q}}[V_{t}|S_{t} = s]},
\end{equation}
where $\sigma_{\text{Dup}}$ is expressed by the Dupire formula, for given call option market prices $C(T,K)$ with maturity time $T$ and strike $K$:
\begin{equation*}
\sigma_{\text{Dup}}^{2}(T,K) = \frac{\partial C(T,K)}{\partial T}\big/\left(\frac{K^2}{2}\cdot\frac{\partial^2 C(T,K)}{\partial K^2}\right),
\end{equation*}
where, for simplicity, we assume zero interest and dividend rates. To obtain the local volatility surface of option prices for all possible strikes and maturities, we need to interpolate and extrapolate the local volatility. As suggested by the authors in \cite{GuyHen}, cubic spline interpolation and flat extrapolation can be employed.\bigbreak

Therefore, it is plausible that a calibrated dynamics of $S$, i.e., an SDE model which is consistent with the observed prices, can be obtained by
reinserting \eqref{Dupireiff} into \eqref{HestonLVM}. 
Equivalently to the initial problem formulation \eqref{HestonLVM}, we consider the following SDE describing the dynamics of the log-spot process $X = \text{log}(S)$ under the risk-neutral measure ${\mathbb{Q}}$:
\begin{equation}\label{sde}
\begin{split}
    &\diff X_t = -\frac{1}{2}V_t\frac{\sigma^2_{\text{Dup}}(t,e^{X_t})}{\mathbb{E}^{\mathbb{Q}}[V_t|X_{t}]}\diff t + \sqrt{V_t}\frac{\sigma_{\text{Dup}}(t,e^{X_t})}{\sqrt{\mathbb{E}^{\mathbb{Q}}[V_{t}|X_{t}]}}\diff W^x_t,\\
    &\diff V_t = k(\theta-V_t)\diff t + \xi\sqrt{V_t}\diff W^v_t,\\
    &\diff W^x_t \diff W^v_t = \rho \diff t, \, \rho \in (-1,1),\\
\end{split}
\end{equation}
with $X_0 \in \mathbb{R}$, $V_0 \in \mathbb{R}^{+}$, and $(W^x, W^v)$ a two-dimensional Brownian motions.\newline 




A key difficulty in the treatment of \eqref{sde}  arises from the conditional expectation $\mathbb{E}^{\mathbb{Q}}[V_{t}|X_{t}]$ that appears in the diffusion coefficient of the calibrated spot price dynamics. This makes the diffusion coefficient dependent on the underlying joint distribution of the state and volatility processes, and therefore a McKean--Vlasov (MV) SDE, but with a non-standard singular measure-dependence. 
As a result, the well-posedness of \eqref{sde} is currently an open question.
Several works deal with special cases and variants (see \cite{Lackeretal,Djete,ZourdainZhou} for versions of the MV SDE and \cite{AbeTach} for an idealized PDE version), but all the results known to us do not cover the problem \eqref{sde} that is precisely of financial interest.\bigbreak

From a practical perspective, simulating the above model poses additional challenges. Several techniques for the simulation of distribution--dependent SDEs have been established over the years. In the context of the calibrated LSV model, one is a PDE approach that is based on solving the Fokker--Planck equation, see, for example, \cite{renetal}. Furthermore, Bayer et al. propose in \cite{BayBelButScho} a novel regularisation approach using reproducing kernel Hilbert space techniques and provide well-posedness and propagation of chaos results for a regularised model. Alternatively, deep learning techniques could be used, as in \cite{Cuchieroetal}, where the authors employ a set of neural networks to parameterize the leverage function. This method allows for model calibration using a generative adversarial network approach, circumventing the need for traditional interpolation methods.
In our work, to approximate the conditional expectation, we use the particle method introduced for a general LSV model by Guyon and Henry-Labord\`ere in \cite{GuyHen}. This leads to an interacting particle system and the following MV SDE with irregular coefficients, which we can simulate.\bigbreak

We take as starting point for our analysis a regularised MV SDE,
\begin{equation}\label{independentparticles}
\begin{split}
\diff X_t &=  \beta(t,(X_t,V_t),\mathbb{P}_{(X_t, V_t)})\diff t + \sigma(t,(X_t,V_t),\mathbb{P}_{(X_t, V_t)})\diff  W^{x}_t, 
\,\,\,t \in [0,T],
\end{split}
\end{equation}
with $V$ as in \eqref{sde} and
\begin{eqnarray*}
\sigma(t,(x,v),\mathbb{P}_{(X_t, V_t)}) = \sqrt{v}\sigma_{\text{Dup}}(t,e^{x})\frac{\sqrt{\mathbb{E}[K(\frac{X-x}{\epsilon})]+ \delta}}{\sqrt{\mathbb{E}[V K(\frac{X-x}{\epsilon})]+\delta}}, \quad
\beta = -\frac{{\sigma}^2}{2},
\end{eqnarray*}
where $K(\cdot)$ is a suitable kernel function, and $\delta$ is a positive constant. 
This can be seen as the limit for an infinite number of particles of the method in \cite{GuyHen}. In fact, one of the results of the present paper is to show exactly that, namely the convergence of the natural particle approximation to the solution of \eqref{independentparticles} (for details, we refer the reader to Section~\ref{particlemethod}).\newline

{
In our earlier work \cite{ReiTsi}, we considered a variant of this model where the volatility process is given by a bounded Lipschitz function of an Ornstein–Uhlenbeck process. We derive explicit Lipschitz constants of the term $\sqrt{}/\sqrt{}$ in $\sigma$ in terms of $\delta$ and $\epsilon$, and can then appeal to standard methods for the well-posedness and numerical approximations of McKean–Vlasov equations. This is not possible in the case of \eqref{independentparticles} due to the presence of the square-roots of the variance process that are a key feature of the Heston model. The results on well-posedness and numerical approximation here are therefore of broader interest in the context of McKean–Vlasov SDEs with non-Lipschitz coefficients. \newline
}

The strong well-posedness of the CIR process is already established (see, e.g., Chapter 5 in \cite{KarShr}). Here, we tackle the strong existence of a unique solution to the log-spot process describing the regularised dynamics. The main difficulty of this task arises from the combination of the H\"{o}lder--1/2 continuity with the measure dependence of the diffusion coefficient, which is also unbounded. Although different works address the well-posedness of MV SDEs with irregular coefficients, none of them covers the specific setting considered in this paper. To name a few, X.\ Erny in \cite{XErny} considers a MV SDE with locally Lipschitz coefficients and proves the strong well-posedness and a propagation of chaos result assuming a bounded diffusion coefficient. 
Another result comes from \cite{LiMaoSongWuYin} where the authors consider a MV SDE with a locally Lipschitz continuous diffusion coefficient in the state variable and bounded state process and prove the well-posedness using an interpolated Euler-like sequence. Furthermore, the authors in \cite{NingJing} prove the strong well-posedness of a class of MV stochastic variational inequalities with measure-dependent coefficients. The above results, however, are not applicable in our setting where the coefficients are unbounded. For further results on the well-posedness of MV SDEs we refer to \cite{Wang2024WellposednessAP, HamSisSzp, HuangWang, ReiSalTug}.\newline

Having established the well-posedness of the regularised MV SDE, we can then prove the convergence of the interacting particle system to the regularised equation as the number of particles $N \to \infty$. In other words, we establish a \textit{propagation of chaos} result. In this paper, we prove strong convergence in the following pathwise sense:
\begin{equation*}
\lim_{N\to \infty}\sup_{i \in \{1,..,N\}} \mathbb{E} \Bigg[\sup_{t \in [0,..,T]} \lvert X^{i,N}_t - X^i_t \rvert^2\Bigg]=0,
\end{equation*} 
where $X^{i,N}$ and $X^{i}$ denote the solutions to the interacting particle system \eqref{particlessystem} and the regularised MV SDE \eqref{independentparticles}, driven by i.i.d.\ copies $(W^{x,i},W^{v,i})$ of the two-dimensional standard Brownian motion $(W^{x},W^{v})$, respectively. This result is well--established under global Lipschitz conditions (see, e.g., \cite{Sznitman} and \cite{Meleard}). In the case of MV SDEs with super-linear drifts, the authors in \cite{ReiEngSmi} prove the above and also provide the rate of convergence. Moreover, the author in \cite{HZhang} provides a propagation of chaos result for a MV SDE with H\"{o}lder continuous coefficient, however, with a diffusion coefficient independent of the measure and therefore not applicable in our setting. For more propagation of chaos results refer to \cite{Carmona, Lacker,XErny}.\bigbreak
In practice, we need to apply a time-discretisation scheme to approximate the particle system. In our work, we use the classic Euler--Maruyama (EM) scheme for the log-spot process and the full-truncation Euler (FTE) scheme \cite{LorKoe} for the volatility process. This ensures the non-negativity of the latter throughout the simulation, and therefore that our scheme is well-defined. Under global Lipschitz conditions, it is well-documented (see, for example, \cite{KloePlat}) that the Euler scheme converges strongly with rate $1/2$ in time. However, our model does not fall into this category due to the CIR volatility process which makes the diffusion coefficient only $1/2$-H\"{o}lder in the second component and not globally Lipschitz continuous. This renders another challenge in the simulation of the H-LSV model. The available literature concerning the propagation of chaos and strong convergence of time-discretisation schemes for MV SDEs with non-Lipschitz coefficients is relatively limited. For an initial but partial investigation into the strong convergence of the EM scheme for MV SDEs with irregular coefficients, we refer to \cite{Zhan}. Bao and Huang in \cite{BaoHua} consider two different cases of MV SDEs with H\"older continuous (i) diffusion and (ii) drift coefficients and prove strong convergence of the Euler--Maruyama scheme and propagation of chaos. Another study on the strong convergence of the EM scheme is conducted by Liu et al.\ in \cite{LiuShiWu}, focusing on the case of super-linear drift and H\"{o}lder continuous diffusion coefficients. However, the aforementioned results consider diffusion coefficients with no mean-field interactions, and are thus not directly applicable to our context. Last but not least, the works \cite{XErny} and \cite{LiMaoSongWuYin} also provide strong convergence results for the EM scheme, however, they are not applicable in our case as explained in the paragraph for the well-posedness above.\bigbreak

{
In summary, the overall numerical solution is affected by four numerical parameters: two regularisation parameters that include the kernel bandwidth $\epsilon$ and an extrapolation parameter $\delta$, the number of particles $N$, and the stepsize of the timestepping scheme $h$. In this work, we fix $\epsilon$ and $\delta$. It is anticipated that as $\epsilon,\delta \rightarrow 0$, the solution of the regularised system \eqref{independentparticles} converges in a suitable sense to the solution of \eqref{sde}, but as the well-posedness of \eqref{sde} is currently unknown and has eluded ourselves as well as prior works outlined above, this is beyond the scope of this work. We refer to \cite{ReiTsi} for a numerical study of how $\epsilon$ and $\delta$ affect the numerical accuracy of the particle and timestepping approximation on the one hand, and the calibration accuracy to the target option prices on the other.\bigbreak

The focus of this work instead is to fill another gap in the literature that comes from the square-root diffusion in \eqref{sde} and  \eqref{independentparticles}, which makes established results for McKean–Vlasov equations inapplicable due to its reduced regularity.
The main results are the following:
}


\begin{enumerate}
\item the well-posedness of the regularised calibrated Heston-LSV model \eqref{independentparticles}, which is a McKean--Vlasov SDE with an unbounded Hölder continuous diffusion coefficient;
\item a propagation of chaos result for the convergence of the particle system to the regularised equation;
\item and the strong convergence of the Euler--Maruyama scheme for the particle system.

\end{enumerate}

\subsection{Structure of paper}

In Section \ref{HCIR}, we review the particle method for approximating the calibrated dynamics and summarise key properties of the CIR process that will be important in our proofs. To keep the focus on the convergence results, we postpone establishing the well-posedness of the regularised equation until Section \ref{WellposednessSection}. In Section \ref{PropOfChaos}, we prove a strong propagation of chaos result, showing that the particle system converges to the regularised MV equation up to a critical time. We then turn to the time-discretised setting in Section \ref{timediscretisation}, where we establish the strong convergence of the Euler--Maruyama scheme applied to the particle system with order $1/2$ in time (up to a logarithmic factor) up to a critical time. Finally, in Section \ref{numexp}, we perform numerical experiments that illustrate our theoretical findings and explore scenarios in which certain assumptions in our proofs may be violated.

\subsection{Notation and Preliminaries}
Given $T>0$, let $(\Omega, \mathcal{F}, \mathbb{P})$ denote a complete probability space endowed with $\mathbb{F} = (\mathcal{F}_t)_{t \in [0,T]}$, the natural filtration that satisfies the usual conditions, i.e., it is right-continuous and $\mathcal{F}_0$ contains all $\mathbb{P}$-null sets. Also, let $(W_t)_{t\in [0,T]}$ be a standard multi--dimensional Brownian motion on $(\Omega, \mathcal{F}, \mathbb{F},\mathbb{P})$ such that it is $\mathbb{F}$-adapted. We use $L^2(\Omega, \mathcal{F},\mathbb{P};\mathbb{R}^n)$ to denote the space of $\mathbb{R}^n-$valued square integrable random variables on $(\Omega, \mathcal{F},\mathbb{P})$ and for any $\theta \in L^2(\Omega, \mathcal{F},\mathbb{P};\mathbb{R}^n),\lVert \theta \rVert_{L^2} := \mathbb{E} \left[ |\theta|^2 \right]^{1/2}$. By $(\mathbb{R}^n,|\cdot|)$ we denote the $n$-dimensional Euclidean space and $|\cdot|$ the Hilbert--Schmidt norm. Also, by $\langle \cdot,\cdot \rangle$ we denote the usual inner product in a given Euclidean space.\bigbreak

\noindent Given a complete separable metric space $(E, d)$ (often being a Polish space), let $\mathcal{P}(E)$ denote the set of probability measures on $(E,\mathcal{B}(E))$, with $\mathcal{B}(E)$ the Borel $\sigma$-field over $E$, and for any $p\ge 1$, let $\mathcal{P}_p(E)$ be the subspace of $\mathcal{P}(E)$ of the probability measures of order $p$, i.e.
 \begin{equation*}
 \mathcal{P}_p(E) := \bigg\{ \mu \in \mathcal{P}(E): \text{ for all }x_0 \in E, \int_E d(x_0,x)^p\mu (\diff x) < + \infty \bigg\},    
 \end{equation*}
 where $d: E \times E \to [0, \infty)$, is a metric of the "distance" between two points in $E$. The set $\mathcal{P}_p(E)$ is equipped with the $p-$Wasserstein distance $\mathcal{W}_p(\mu,\nu)$ defined as
\begin{equation}
\mathcal{W}_p(\mu,\nu) = \inf_{\pi}\left(\iint_{E \times E}d(x,y)^p \pi (\diff x,\diff y) \right)^{\frac{1}{p}},
\end{equation}
where $\pi \in \mathcal{P}(E\times E)$ such that $\pi(. \times E) = \mu $ and $\pi(E \times .) = \nu$, i.e. $\pi$ is a coupling of $\mu$ and $\nu$, for any $p\ge 1$ and $\mu, \nu \in \mathcal{P}_p(E).$ Here we work on $\mathbb{R}^n$, and $d$ is the standard Euclidean norm.

\section{Problem formulation}\label{HCIR}

We recall for convenience the complete regularised MV-SDE, with $Z=(X,V)$, $z=(x,v)$, $\mu^{Z}_t = \mathbb{P}_{(X_t, V_t)}$:
\begin{equation}\label{independentparticles3}
\begin{split}
\diff X_t &=  \beta(t,Z_t,\mu^{Z}_t)\diff t + \sigma(t,Z_t,\mu^{Z}_t)\diff  W^{x}_t, 
\,\,\,t \in [0,T], \\
\diff V_t &= k(\theta-V_t)\diff t + \xi\sqrt{V_t}\diff W^v_t,
\end{split}
\end{equation}
with $\diff W^x_t \diff W^v_t = \rho \diff t$, $\rho \in (-1,1)$, and
\begin{eqnarray*}
\sigma(t,z,\mu^{Z}_t) = \sqrt{v}\sigma_{\text{Dup}}(t,e^{x})
\frac{\sqrt{\mathbb{E}^{\mu^{Z}_t}[K(\frac{X-x}{\epsilon})]+ \delta}}{\sqrt{\mathbb{E}^{\mu^{Z}_t}[V K(\frac{X-x}{\epsilon})]+\delta}}, \quad
\beta = -\frac{{\sigma}^2}{2}.
\end{eqnarray*}

For the analysis, we make the following assumptions on the coefficient functions.
\begin{assumption}\label{A3} We assume that
\begin{enumerate}[label=(\subscript{A}{{\arabic*}})]
    \item $(t,x) \to \sigma_{\mathrm{Dup}}(t,e^{x})$ is bounded, Lipschitz in $x$, and $\frac{1}{2}$-H\"{o}lder  in $t$, so that there exist positive constants $A_1, \,L_{\mathrm{Dup}}$ such that for all $x_1,x_2 \in \mathbb{R}$, and $t_1,t_2 \in [0,T],$\newline
$|\sigma_{\mathrm{Dup}}(t_1,e^{x_1})|\le A_1, \text{ and } \big|\sigma_{\mathrm{Dup}}(t_1,e^{x_1}) - \sigma_{\mathrm{Dup}}(t_2,e^{x_2})\big| \le L_{\mathrm{Dup}}\big(|t_1-t_2|^{1/2} + |x_1 - x_2| \big).$ 
    \item $K(\cdot)$ is bounded and Lipschitz continuous so that there exist positive constants $A_2, L_K$ such that for all $x_1,x_2 \in \mathbb{R}$, 
    $|K(x_1)|\le A_2, $ and  $|K(x_1)-K(x_2)|\le L_K |x_1-x_2|$.
\end{enumerate}
\end{assumption}

Neither of these assumptions is practically very restrictive. For the function $L_{\mathrm{Dup}}$, a parametric ansatz is typically calibrated to a discrete set of options data. Common choices are splines in the spatial (i.e., second) variable with some extrapolation, which can be made to satisfy the assumption.
The assumption could be weakened to allow a finite number of discontinuities in time (i.e., the first variable) without complicating the analysis.

\begin{subsection}{Particle method}\label{particlemethod}
To approximate the conditional expectation that appears in the calibrated dynamics of the log-spot process $X$ in \eqref{sde}, we apply a particle method as introduced by Guyon and Henry--Labord\`ere in \cite{GuyHen} for a general LSV model. This involves approximating the true measure $\mu^{Z}_t$ of the joint law of $(X_t,V_t)$ by
\begin{equation*}
 \mu^{Z^N}_t(x,v; \epsilon_x,\epsilon_v) := \frac{1}{N}\sum_{i=1}^{N}\Phi_{\epsilon_x}(X_t^{i,N}-x)\Phi_{\epsilon_v}(V_t^{i}-v),
\end{equation*}
where $\Phi_{\epsilon}(\cdot)$ is a regularising kernel function of the form 
\begin{equation} \label{ReTskernel}
    \Phi_{\epsilon}(x) = \epsilon^{-1}K\left(\epsilon^{-1} x\right),
\end{equation}
with $K(\cdot)$ a real--valued, non--negative kernel function satisfying $\int_{-\infty}^{+\infty} K(u)\diff u = 1$, and $\epsilon$ the bandwidth of the kernel. For ease of notation, from now on, we denote $\epsilon_x$ by $\epsilon$. This leads to the Nadaraya--Watson estimator for the conditional expectation as follows:
\begin{equation*}
\mathbb{E}[V_t|X_t=x] \approx \frac{\frac{1}{N}\sum_{i=1}^{N}V_t^{i}\Phi_{\epsilon}(X_t^{i,N}-x)}{\frac{1}{N}\sum_{i=1}^{N}\Phi_{\epsilon}(X_t^{i,N}-x)},
\end{equation*}
where $(\textbf{X}^N_t)_{t\in[0,T]}:= \big(X_t^{1,N}, X_t^{2,N}, ... , X_t^{N,N}\big)^\intercal_{t\in[0,T]} $ denotes the interacting particle system for $X$, and $(\textbf{V}_t)_{t\in[0,T]}:= \big(V_t^{1}, V_t^{2}, ... , V_t^{N}\big)^\intercal_{t\in[0,T]} $ independent Monte Carlo samples for $V$.\newline
$\mathbb{E}^{\mu_t^{Z}}\left[V_t\Phi_{\epsilon}(X_t-x)\right]$ is then estimated by
\begin{equation*}\label{NW22}
 \mathbb{E}^{\mu_t^{Z^N}} \left[V_t\Phi_{\epsilon}(X_t-x)\right] = \iint \left( \tilde{V}\Phi_{\epsilon}(\tilde{X}-x)\right)\diff \mu_t^{Z^N} ( \tilde{X},\tilde{V})= \frac{1}{N} \sum_{i=1}^N\,{V^i}\,\Phi_{\epsilon}(X^{i,N}-x).
\end{equation*}
This leads to the following system that we can simulate.\bigbreak

\noindent \textit{Interacting particle system}\bigbreak

\noindent Consider particles $X^{i,N}$, $i=1,...,N$, satisfying the following SDE,
\begin{equation}\label{particlessystem}
\diff X^{i,N}_t = \beta(t,(X_t^{i,N},V_t^{i}), \mu_t^{Z^N})\diff t+\sigma(t,(X_t^{i,N},V_t^{i}),\mu_t^{Z^N})\diff{W^{x,i}_t},\, t\in [0,T],
\end{equation}
with i.i.d. $X^{i,N}_0=X^{i}_0$, and
\begin{equation*}
\sigma(t,(X_t^{i,N},V_t^{i}), \mu_t^{Z^N}) = \sqrt{V_t^{i}}\,\sigma_{\text{Dup}}(t,e^{X^{i,N}_t})\frac{\sqrt{\sum_{j=1}^{N}K\big(\frac{X_t^{j,N} - X_t^{i,N}}{\epsilon}\big)+\delta}}{\!\!\!\!\sqrt{\sum_{j=1}^{N}V_t^{j}K\big(\frac{X_t^{j,N} - X_t^{i,N}}{\epsilon}\big)+\delta}}, \qquad \beta = -\frac{\sigma^2}{2}.
\end{equation*}
The positive constant $\delta$ is added to the numerator and denominator to avoid potential singularities. Also, for each $i$, $W^{x,i}$ are independent Brownian motions and $V^i$ independent Monte Carlo samples that evolve according to the dynamics
\begin{equation}\label{volparticles}
\diff V^{i}_t = k(\theta - V^{i}_t)\diff t + \xi \sqrt{V_t^i} \diff W^{v}_t, \,\,\, V^i_0 = V^i_0, \,\,\, t\in[0,T].
\end{equation}
The interaction term $\mu_t^{Z^N}$ distinguishes the particle method from the classical Monte Carlo method since the paths in the former are no longer independent. As discussed above, it is necessary for the particle system to satisfy a \textit{propagation of chaos} property. 
We refer to the Introduction, where we discuss some of the available literature on propagation of chaos results for irregular MV SDEs. Since we are interested in the strong convergence, we prove a strong convergence result in Theorem \ref{PropChaos}. In the limit, we expect the particles to become independent and satisfy the following MV SDE, which we restate from the Introduction.\bigbreak

\noindent \textit{Non--interacting particle system}\bigbreak

\noindent For each $i$, let $X^i$ be independent particles that satisfy the MV SDE
\begin{equation}\label{independentparticles2}
\begin{split}
\diff X^i_t &=  \beta(t,(X^i_t,V^i_t),\mu^{Z}_t)\diff t + \sigma(t,(X^i_t,V^i_t),\mu^{Z}_t)\diff  W^{x,i}_t,\,\,\, X^i_0=X^i_0,\,\,\,t \in [0,T],
\end{split}
\end{equation}
with 
\begin{eqnarray*}
\sigma(t,(x,v),\mu^Z) = \sqrt{v}\sigma_{\text{Dup}}(t,e^{x})\frac{\sqrt{\mathbb{E}^{\mu^Z}[K(\frac{X-x}{\epsilon})]+ \delta}}{\sqrt{\mathbb{E}^{\mu^Z}[V K(\frac{X-x}{\epsilon})]+\delta}}, \quad
\beta = -\frac{{\sigma}^2}{2}.
\end{eqnarray*}
Here, $\mu^{Z}_t$ denotes the true joint law of $(X^i_t,V^i_t)$, which is the same for all $(X^i_t,V^i_t)$ since the particles here are independent.\newline

\begin{remark}
For convenience of notation, let $\tilde{\beta}(t,x,\nu) := -\frac{1}{2}\sigma^2_{\text{Dup}}(t,e^{x})\frac{\mathbb{E}^{\nu}[K(\frac{X-x}{\epsilon})]+ \delta}{\mathbb{E}^{\nu}[V K(\frac{X-x}{\epsilon})]+\delta},$ and $\tilde{\sigma}(t,x,\nu) := \sigma_{\text{Dup}}(t,e^{x})\frac{\sqrt{\mathbb{E}^{\nu}[K(\frac{X-x}{\epsilon})]+ \delta}}{\sqrt{\mathbb{E}^{\nu}[V K(\frac{X-x}{\epsilon})]+\delta}}$, so that $\beta(t,(x,v),\nu) := v \tilde{\beta}(t,x,\nu)$ and $\sigma(t,(x,v),\nu) := \sqrt{v} \tilde{\sigma}(t,x,\nu)$. Under Assumption \ref{A3}, and by adapting Proposition 1 in \cite{ReiTsi} to system \eqref{sde}, we conclude that $\tilde{\beta}(t,x,\nu)$ and $\tilde{\sigma}(t,x,\nu)$ are Lipschitz continuous in the state and measure variables with Lipschitz constants $L_{\tilde{\beta}}$ and $L_{\tilde{\sigma}}$ respectively, and $1/2-$H\"older continuous in time. $L_{\tilde{\beta}}$ and $L_{\tilde{\sigma}}$ depend on $\epsilon, \delta$ and the constants in Assumption \ref{A3}. For the exact dependence, we refer to the proof of Proposition 1 in \cite{ReiTsi}. In our work, we keep $\epsilon, \delta$, and the constants in Assumption \ref{A3} fixed. Also, $|\tilde{\beta}(t,x,\nu)|$ and $|\tilde{\sigma}(t,x,\nu)|$ are bounded by positive constants $\tilde{\beta}_{\max}$ and $\tilde{\sigma}_{\max}$, respectively, which depend on $A_1$, $A_2$, and $\delta$. It holds that, $\tilde{\beta}_{\max} = \frac{\tilde{\sigma}^2_{\max}}{2}$.
\end{remark}

\begin{remark}
In the proofs below, we use a technique involving the stopping--time $\tau^i_{\omega}:= \underset{t\ge 0}{\inf}\{\int_0^t V^i_u \diff u \ge \omega\}$. Although using stopping--times for McKean--Vlasov SDEs is tricky to handle, this is not the case here since the stopped process follows a standard SDE.
\end{remark}

\subsection{Properties of the Cox--Ingersoll--Ross process}
Before establishing the convergence of the particle system, we recall key properties of the CIR process that play a crucial role in our analysis. For further details, see for instance \cite{AndPit}, \cite{CozMarRei}, \cite{CozRei}, \cite{DNS}, \cite{HurKuz}, and the references therein.\bigbreak

\noindent Let $\nu = \frac{2k\theta}{\xi^2}$ be the Feller ratio and $\nu^* = 2+\sqrt{3}$. Also, for $\lambda >0$, let 
\begin{equation}\label{Tbound}
T^* :=
\begin{cases} 
\frac{2}{\sqrt{2\lambda \xi^2 -k^2}}\left[ \frac{\pi}{2} +\text{tan}^{-1}\left(\frac{k}{\sqrt{2\lambda \xi^2 -k^2}} \right)\right],& \text{if } k^2  < 2 \lambda \xi^2 , \\
\infty, & \text{if } k^2  \geq 2 \lambda \xi^2.
\end{cases}
\end{equation}
\begin{lemma}\label{momentboundvol}
  Let $V^i$ denote the CIR volatility process as in \eqref{volparticles}. It holds that 
\begin{enumerate}
\item $\underset{t\in [0,T]}{\sup}\mathbb{E}\left[{V_t^i}^p\right] < \infty, \, \forall \, p>-\nu, \, \textit{i.e. it has bounded moments,}$
\item $     \mathbb{E}\left[\underset{t\in [0,T]}{\sup}{V_t^i}^p\right] < \infty, \, \forall \, p\ge 1, \, \textit{i.e. it has uniformly bounded moments}.$
\end{enumerate}
\end{lemma}
\begin{proof}
For the first bound we refer the reader to formula $(3.1)$ in \cite{DNS}, while for a proof of the second bound to Proposition 3.7 in \cite{CozMarRei}.
\end{proof}
\begin{lemma}\label{ExpIntVol}
Let $V^i$ denote the CIR volatility process as in \eqref{volparticles} and $\lambda >0$. For $T<T^*$, with $T^*$ the critical time as in \eqref{Tbound}, we have that
\begin{equation}
   \mathbb{E}\left[\exp\{\lambda \int_0^T {V_t}^i\diff t\}\right]< \infty.
\end{equation}
\end{lemma}
\begin{proof}
 This follows from Proposition 3.1 in \cite{AndPit}.  
\end{proof}

\subsection{Well-posedness of the regularised system}
For completeness, we state the existence and uniqueness result for equation \eqref{independentparticles3} here, while deferring its proof and analysis to Section~\ref{WellposednessSection} to maintain the focus on the convergence analysis.
\begin{theorem}[Well-posedness]\label{wellposedness}
Let Assumptions \ref{A3} hold and also assume that $\nu \ge 1$ and $\mathbb{E}[|X_0|^4] < \infty$. Then, there exists a unique strong solution to equation \eqref{independentparticles3} up to a critical time $T^*$, for $T^*$ as in $\eqref{Tbound}$.
\end{theorem}

 We are now well-equipped to prove a strong propagation of chaos result that establishes the convergence of the particle system \eqref{particlessystem} to the McKean--Vlasov SDE \eqref{independentparticles2},
 and in the following section a timestepping approximation scheme. \newline

\section{Propagation of chaos}\label{PropOfChaos}
 To prove the result below we use techniques from \cite{ReiEngSmi} and \cite{CozRei2}. The authors in \cite{ReiEngSmi} prove propagation of chaos for a MV SDE with drift of super-linear growth and also provide the rate of convergence. However, their analysis assumes a global Lipschitz diffusion in the state variable. The study in \cite{CozRei2} addresses the square--root diffusion by considering a Heston-type LSV model and proves strong convergence of the Euler scheme. However, they do not consider the calibrated system so their equation is a standard SDE, without mean-field dependence.
\begin{theorem}\label{PropChaos}
    Let $Z^i := (X^i,V^i)$ and $Z^{i,N} := (X^{i,N},V^i)$ with $X^i$, $X^{i,N}$, and $V^i$ solutions to equations \eqref{independentparticles2}, \eqref{particlessystem}, and \eqref{volparticles} respectively. Further, let $Z^{i}_0,Z^{i,N}_0 \in L^2(\mathbb{R}^2)$. Suppose that  $\nu \ge 1$ and that Assumption \ref{A3} holds. For $T<T^*$ we have that,
\begin{equation*}
\lim_{N\to \infty}\sup_{i \in \{1,..,N\}} \mathbb{E} \Bigg[\sup_{t \in [0,T]} \lvert Z^{i,N}_t - Z^i_t \rvert^2\Bigg]=0.
\end{equation*} 
\end{theorem}
\begin{proof}
In our proof, we let $C$ denote a constant independent of $N$, and only possibly dependent on $T$ and the Lipschitz constants $L_{\tilde{\beta}}$ and $L_{\tilde{\sigma}}$ which depend on $\epsilon$ and $\delta$. We allow $C$ to vary from line to line.\bigbreak

Note that $\lvert Z^{i,N}_t - Z^i_t \rvert^2 =  \lvert X^{i,N}_t - X^i_t \rvert^2$ since the second component of $Z^i$ and $Z^{i,N}$ is identical for all $i$. We fix $i \in \{1,...,N\}$ and let $E^N_t:=X^{i,N}_t - X^i_t,\,\,E^N_0=0 $. Applying It\^{o}'s formula to $|E^N_{t \wedge \tau}|^{2p}$, for $\tau$ a stopping time and $1<p<2$, we get
\begin{equation*}
\begin{split}
|E^N_{t \wedge \tau}|^{2p} &= 2p \int_0^{t \wedge \tau} |E^N_{u}|^{2p-1}\text{sgn}(E^N_{u})(\beta(u,(X_u^{i,N},V^i_u), \mu_u^{Z^N})-\beta(u,(X_u^{i},V^i_u), \mu_u^{Z})) \diff u \\
&+2p \int_0^{t \wedge \tau} |E^N_{u}|^{2p-1}\text{sgn}(E^N_{u})(\sigma(u,(X_u^{i,N},V^i_u), \mu_u^{Z^N})-\sigma(u,(X_u^{i},V^i_u), \mu_u^{Z})) \diff  W^{x,i}_u\\
&+\frac{1}{2}2p(2p-1)\int_0^{t \wedge \tau}|E^N_{u}|^{2p-2}\lvert\sigma(u,(X_u^{i,N},V^i_u), \mu_u^{Z^N})-\sigma(u,(X_u^{i},V^i_u), \mu_u^{Z})\lvert^2 \diff u, \\ 
\end{split}
\end{equation*}
where $\text{sgn}(E^N_u)=1$ if $E^N_u>0$ and $\text{sgn}(E^N_u)=-1$ otherwise.\newline
Taking the supremum over $t \in [0,T]$ on both sides, noting that the supremum of the integrals occurs at $t=T$, and then expectations yields
\begin{equation}\label{propsplit1}
    \begin{split}
        &\mathbb{E}\left[\sup_{t\in[0,T]} |E^N_{t \wedge \tau}|^{2p}\right] \le 2p  \mathbb{E}\left[\int_0^{T \wedge \tau} |E^N_{u}|^{2p-1}\lvert\beta(u,(X_u^{i,N},V^i_u), \mu_u^{Z^N})-\beta(u,(X_u^{i},V^i_u), \mu_u^{Z})\lvert \diff u\right] \\
        &+2p \mathbb{E}\left[\sup_{t\in[0,T]}\int_0^{t \wedge \tau} |E^N_{u}|^{2p-1}\text{sgn}(E^N_{u})(\sigma(u,(X_u^{i,N},V^i_u), \mu_u^{Z^N})-\sigma(u,(X_u^{i},V^i_u), \mu_u^{Z})) \diff  W^{x,i}_u\right]\\
        &+\frac{1}{2}2p(2p-1)\mathbb{E}\left[\int_0^{T \wedge \tau}|E^N_{u}|^{2p-2}\lvert\sigma(u,(X_u^{i,N},V^i_u), \mu_u^{Z^N})-\sigma(u,(X_u^{i},V^i_u), \mu_u^{Z})\lvert^2 \diff u\right].\\ 
    \end{split}
\end{equation}
We treat each term on the RHS above separately. Looking at the first term of \eqref{propsplit1}, 
\begin{equation*}
\begin{split}
  |E^N_{u}|^{2p-1}&V^i_u\lvert\tilde{\beta}(u,X_u^{i,N}, \mu_u^{Z^N})-\tilde{\beta}(u,X_u^{i},\mu_u^{Z})\lvert \le L_{\tilde{b}} |E^N_{u}|^{2p-1}V^i_u \Big( \lvert E^N_{u} \lvert + \mathcal{W}_2(\mu_u^{Z^N},\mu_u^{Z} )\Big)\\
  &\le L_{\tilde{b}} |E^N_{u}|^{2p-1} V^i_u \Big( |E^N_{u}| + \mathcal{W}_2(\mu_u^{Z^N},\mu_u^{N}) +\mathcal{W}_2(\mu_u^{N},\mu_u^{Z} )\Big),
\end{split} 
\end{equation*}
where $\mu_u^{N}$ denotes the empirical measure of $Z^{i}_u$ and we applied the property of the Wasserstein metric, see Chapter 6 in \cite{Vil}, that for $\mu_1,\mu_2,\mu_3 \in \mathcal{P}_2(\mathbb{R}^{2n}), \, \mathcal{W}_2(\mu_1,\mu_2) \le \mathcal{W}_2(\mu_1,\mu_3)+\mathcal{W}_2(\mu_3,\mu_2).$\newline
Also, since $\mu_u^{Z^N}$ and $\mu_u^{N}$ are empirical measures we have the standard bound
\begin{equation*}
\mathcal{W}_2(\mu_u^{Z^N},\mu_u^{N})\le \left( \frac{1}{N}\sum_{j=1}^N |Z^{j,N}_u-Z^{j}_u|^2 \right)^{1/2} =\left( \frac{1}{N}\sum_{j=1}^N |X^{j,N}_u-X^{j}_u|^2 \right)^{1/2}.   
\end{equation*}
Also, by Young's inequality one gets
\begin{equation*}
    \begin{split}
    |E^N_{u}|^{2p-1}\mathcal{W}_2(\mu_u^{N},\mu_u^{Z}) &\le \frac{2p-1}{2p} |E^N_{u}|^{2p} +\frac{1}{2p} \mathcal{W}_2(\mu_u^{N},\mu_u^{Z})^{2p},\\
    |E^N_{u}|^{2p-1}\left( \frac{1}{N}\sum_{j=1}^N |X^{j,N}_u-X^{j}_u|^2 \right)^{1/2} & \le \frac{2p-1}{2p}|E^N_{u}|^{2p} +\frac{1}{2p}\left( \frac{1}{N}\sum_{j=1}^N |X^{j,N}_u-X^{j}_u|^2 \right)^{p}\\
    & \le \frac{2p-1}{2p}|E^N_{u}|^{2p} + \frac{1}{2p}N^{-1}\sum_{j=1}^N |X^{j,N}_u-X^{j}_u|^{2p},\\
    \end{split}
\end{equation*}
where the last step follows by applying H\"{o}lder's inequality to the sum.
We therefore have
\begin{equation*}
\begin{split}
    &\mathbb{E}\left[ \int_0^{T \wedge \tau}|E^N_{u}|^{2p-1}\lvert\beta(u,(X_u^{i,N},V^i_u), \mu_u^{Z^N})-\beta(u,(X_u^{i},V^i_u), \mu_u^{Z})\lvert \diff u \right]\\
    &\le C \mathbb{E}\left[ \int_0^{T \wedge \tau}\Big( |E^N_{u}|^{2p} +N^{-1}\sum_{j=1}^N |X^{j,N}_u-X^{j}_u|^{2p}+ \mathcal{W}_2(\mu_u^{N},\mu_u^{Z})^{2p} \Big)V^i_u\diff u\right].\\   
\end{split}
\end{equation*}
Notice that by exchangeability of the components of $X$ and $X^N$, 
\begin{equation*}
  \mathbb{E}\left[N^{-1}\sum_{j=1}^N |X^{j,N}_u-X^{j}_u|^{2p} \right] = \mathbb{E}\left[|X^{i,N}_u-X^{i}_u|^{2p} \right]= \mathbb{E}\left[|E^N_{u}|^{2p} \right].
\end{equation*}
We now define the following stochastic process which is strictly increasing since $V^i$ has almost surely strictly positive paths for $\nu \ge 1$:
\begin{equation}\label{g}
    g^i(t) = \int_0^tV^i_u \diff u.
\end{equation}
Applying Fubini's theorem yields
\begin{equation}
    \begin{split}
       &\mathbb{E}\left[ \int_0^{T \wedge \tau}|E^N_{u}|^{2p-1}\lvert\beta(u,(X_u^{i,N},V^i_u), \mu_u^{Z^N})-\beta(u,(X_u^{i},V^i_u), \mu_u^{Z})\lvert \diff u \right]\\
       & \le C \int_0^{T \wedge \tau} \Big(\mathbb{E}\left[\sup_{t \in [0,u]} |E^N_{t}|^{2p}\right]+\mathbb{E}\left[\mathcal{W}_2(\mu_u^{N},\mu_u^{Z})^{2p}\right]\Big) \diff g^i(u)\\
       & \le T_1:= C  \int_0^{T \wedge \tau}\left(\mathbb{E}\left[\sup_{t \in [0,T]} |E^N_{t\wedge u}|^{2p}\right] +\mathbb{E}\left[\mathcal{W}_2(\mu_u^{N},\mu_u^{Z})^{2p}\right]\right) \diff g^i(u).\\
    \end{split}
\end{equation}
Now looking at the second term of equation \eqref{propsplit1}, to remove the stochastic integral one may apply the Burkholder--Davis--Gundy (BDG) inequality to get for $C$ a positive constant,
\begin{equation}
\begin{split}\label{propsplit2}
&2p\mathbb{E}\left[\sup_{t\in[0,T]}\int_0^{t \wedge \tau} |E^N_{u}|^{2p-1}\text{sgn}(E^N_{u})(\sigma(u,(X_u^{i,N},V^i_u), \mu_u^{Z^N})-\sigma(u,(X_u^{i},V^i_u), \mu_u^{Z})) \diff  W^{x,i}_u\right] \\
 & \le C\mathbb{E}\left[\left(\int_0^{T \wedge \tau} |E^N_{u}|^{2(2p-1)}|\sigma(u,(X_u^{i,N},V^i_u), \mu_u^{Z^N})-\sigma(u,(X_u^{i},V^i_u), \mu_u^{Z})|^2 \diff u\right)^{\frac{1}{2}}\right]\\ 
  & \le \mathbb{E}\left[\left(\sup_{t\in [0,T]}|E^N_{t \wedge \tau}|^{2p}C^2\int_0^{T \wedge \tau} |E^N_{u}|^{2p-2}|\sigma(u,(X_u^{i,N},V^i_u), \mu_u^{Z^N})-\sigma(u,(X_u^{i},V^i_u), \mu_u^{Z})|^2 \diff u\right)^{\frac{1}{2}}\right]\\ 
  &\le \frac{1}{2}\mathbb{E}\left[\sup_{t\in [0,T]}|E^N_{t \wedge \tau}|^{2p}\right]+\frac{1}{2}C^2\mathbb{E}\left[\int_0^{T \wedge \tau} |E^N_{u}|^{2p-2}|\sigma(u,(X_u^{i,N},V^i_u), \mu_u^{Z^N})-\sigma(u,(X_u^{i},V^i_u), \mu_u^{Z})|^2 \diff u\right],\\ 
\end{split}
\end{equation}
which follows from the arithmetic and geometric mean (AM--GM) inequality. It is left to find a bound for the last term in equation \eqref{propsplit1}.
\begin{equation*}
\begin{split}
& \mathbb{E}\left[\int_0^{T \wedge \tau}|E^N_{u}|^{2p-2}\lvert\sigma(u,(X_u^{i,N},V^i_u), \mu_u^{Z^N})-\sigma(u,(X_u^{i},V^i_u), \mu_u^{Z})\lvert^2 \diff u\right]\\
 &= \mathbb{E}\left[\int_0^{T \wedge \tau}|E^N_{u}|^{2p-2}V^i_u\lvert\tilde{\sigma}(u,X_u^{i,N}, \mu_u^{Z^N})-\tilde{\sigma}(u,X_u^{i}, \mu_u^{Z})\lvert^2 \diff u\right]\\
 &\le L_{\tilde{\sigma}}^2 \mathbb{E}\left[\int_0^{T \wedge \tau}|E^N_{u}|^{2p-2}V^i_u \big(|E^N_{u}| + \mathcal{W}_2(\mu_u^{Z^N},\mu_u^{Z}) \big)^2 \diff u\right]\\
&\le T_2:= L_{\tilde{\sigma}}^2 \mathbb{E}\left[\int_0^{T \wedge \tau}|E^N_{u}|^{2p-2}V^i_u \big(2|E^N_{u}|^2 + 4\mathcal{W}^2_2(\mu_u^{Z^N},\mu_u^{N}) + 4\mathcal{W}^2_2(\mu_u^{N},\mu_u^{Z}) \big) \diff u\right],\\
\end{split}
\end{equation*}
where $\mu_u^{N}$ is the empirical measure of $Z^{i}_u$ as above. Recall the definition of process $g^i$ in \eqref{g}. We then have for $C$ a positive constant that depends on $p$ and is of order $L_{\tilde{\sigma}}^2$,
\begin{equation*}
    \begin{split}
        &T_2 := L_{\tilde{\sigma}}^2 \mathbb{E}\left[\int_0^{T \wedge \tau}|E^N_{u}|^{2p-2}\big(2|E^N_{u}|^2 + 4\mathcal{W}^2_2(\mu_u^{Z^N},\mu_u^{N}) + 4\mathcal{W}^2_2(\mu_u^{N},\mu_u^{Z}) \big) \diff g^i(u)\right]\\
        & \le C \mathbb{E}\left[\int_0^{T \wedge \tau}\big(|E^N_{u}|^{2p} +|E^N_{u}|^{2p-2}\mathcal{W}^2_2(\mu_u^{Z^N},\mu_u^{N}) + |E^N_{u}|^{2p-2}\mathcal{W}^2_2(\mu_u^{N},\mu_u^{Z}) \big) \diff g^i(u)\right]\\
        & \le C \mathbb{E}\left[\int_0^{T \wedge \tau}\Big(|E^N_{u}|^{2p}+\frac{2(p-1)}{p}|E^N_u|^{2p}+\frac{1}{p}\mathcal{W}^{2p}_{2}(\mu_u^{Z^N}, \mu_u^{N})+\frac{1}{p}\mathcal{W}^{2p}_{2}(\mu_u^{N}, \mu_u^{Z}) \Big) \diff g^i(u)\right]\\
        & \le C \mathbb{E}\left[\int_0^{T \wedge \tau}\Big(|E^N_{u}|^{2p}+N^{-1}\sum_{j=1}^N |X^{j,N}_u-X^{j}_u|^{2p}+\mathcal{W}^{2p}_{2}(\mu_u^{N}, \mu_u^{Z}) \Big) \diff g^i(u)\right]\\
        & \le C \mathbb{E}\left[\int_0^{T \wedge \tau}\left(\sup_{t\in[0,u]}|E^N_{t}|^{2p}+N^{-1}\sum_{j=1}^N \sup_{t\in[0,u]}|X^{j,N}_t-X^{j}_t|^{2p}+\mathcal{W}^{2p}_{2}(\mu_u^{N}, \mu_u^{Z}) \right) \diff g^i(u)\right]\\
        & \le T_3 := C \mathbb{E}\left[\int_0^{T \wedge \tau}\left(\sup_{t\in[0,T]}|E^N_{t \wedge u}|^{2p}+N^{-1}\sum_{j=1}^N \sup_{t\in[0,T]}|X^{j,N}_{t\wedge u}-X^{j}_{t\wedge u}|^{2p}+\mathcal{W}^{2p}_{2}(\mu_{u}^{N}, \mu_{ u}^{Z}) \right) \diff g^i(u)\right].\\
    \end{split}
\end{equation*}
Note that by exchangeability, 
\begin{equation*} 
  \mathbb{E}\left[N^{-1}\sum_{j=1}^N \sup_{t\in[0,T]}|X^{j,N}_{t\wedge u}-X^{j}_{t\wedge u}|^{2p}\right] = \mathbb{E}\left[\sup_{t\in[0,T]}|X^{j,N}_{t\wedge u}-X^{j}_{t\wedge u}|^{2p}\right]= \mathbb{E}\left[\sup_{t\in[0,T]}|E^N_{t\wedge u}|^{2p} \right].
\end{equation*}
Adopting a technique from \cite{BerBosDio} and \cite{CozRei2}, we let $\tau^i_{\omega}$ be a stopping time such that $\tau^i_{\omega}:= \inf\{t\ge 0 \, \lvert \, g^i(t) \ge \omega\}$, for any $\omega \ge 0$ and $\tau^i_0 = 0$. It holds that $g^i(\tau^i_{\omega}) = \omega$ and $g^i(T \wedge \tau^i_{\omega}) = g^i(T)\wedge g^i(\tau^i_{\omega}) = g^i(T)\wedge \omega$.
\bigbreak
We now fix $\omega >0$, set $\tau = \tau^i_{\omega}$, and consider a stochastic time change $s= g^i(u)$ so that $u = \tau^i_s$. Notice that $s$ also depends on $i$. From now on, we drop the superscript $i$ from $\tau^i_{\cdot}$ and $g^i(\cdot)$ for ease of notation. Applying the Lebesgue's change of time formula and Fubini's theorem gives
\begin{equation}\label{propsplit3}
\begin{split}
T_1 +T_3 &\le C \mathbb{E}\left[\int_0^{g(T) \wedge \omega}\left(\sup_{t\in[0,T]}|E^N_{t \wedge {\tau}_s}|^{2p}+\mathcal{W}^{2p}_{2}(\mu_{{\tau}_s}^{N}, \mu_{ {\tau}_s}^{Z}) \right) \diff s\right]\\
& \le C \int_0^{\omega}\left(\mathbb{E}\left[\sup_{t\in[0,T]}|E^N_{t \wedge {\tau}_s}|^{2p}\right]+\mathbb{E}\left[\mathcal{W}^{2p}_{2}(\mu_{{\tau}_s}^{N}, \mu_{ {\tau}_s}^{Z}) \right]\right) \diff s.\\
\end{split}
\end{equation}
Putting equations \eqref{propsplit1}, \eqref{propsplit2}, and \eqref{propsplit3} together we have
\begin{equation}\label{propsplit4}
    \begin{split}
        &\mathbb{E}\left[\sup_{t\in[0,T]} |E^N_{t \wedge \tau_{\omega}}|^{2p}\right] \le C \int_0^{\omega}\mathbb{E}\left[\mathcal{W}^{2p}_{2}(\mu_{{\tau}_s}^{N}, \mu_{ {\tau}_s}^{Z}) \right]\diff s +C\int_0^{\omega}\mathbb{E}\left[\sup_{t\in[0,T]}|E^N_{t \wedge {\tau}_s}|^{2p}\right]\diff s.\\ 
        \end{split}
\end{equation}
It follows by Gr\"{o}nwall's inequality that for $\omega >0$,
\begin{equation}\label{gronwalls1}
  \mathbb{E}\left[\sup_{t\in[0,T]} |E^N_{t \wedge \tau_{\omega}}|^{2p}\right] \le   C  e^{C\omega}\int_0^{\omega}\mathbb{E}\left[\mathcal{W}^{2p}_{2}(\mu_{{\tau}_s}^{N}, \mu_{ {\tau}_s}^{Z}) \right]\diff s.
\end{equation}
By Lemma 1.9 in \cite{Carmona} we have that $\mathcal{W}^{2}_{2}(\mu_{{\tau}_s}^{N}, \mu_{ {\tau}_s}^{Z})$ converges to $0$ as $N\to \infty$ almost surely. Building on this result and by the estimate in Lemma \ref{aprioriestimates}, the sequence $\left(\mathcal{W}^{2p}_{2}(\mu_{{\tau}_s}^{N}, \mu_{ {\tau}_s}^{Z})\right)_{N \ge 1}$ of random variables is uniformly integrable and $\mathcal{W}^{2p}_{2}(\mu_{{\tau}_s}^{N}, \mu_{ {\tau}_s}^{Z})$ converges to $0$ almost surely for $\tau_s <T^*$. By the Vitali convergence theorem, the convergence also holds in the sense of $L^1$. We therefore conclude that for all $\tau_s <T^*$, 

\begin{equation}\label{Wassbound3}
\underset{N \to \infty}{\lim}\mathbb{E}\left[\mathcal{W}^{2p}_{2}(\mu_{{\tau}_s}^{N}, \mu_{ {\tau}_s}^{Z}) \right]=0.
\end{equation}
Let $u = \tau_{\omega}$ for ease of notation. Therefore, by taking limits in \eqref{gronwalls1} and applying Fubini's theorem, we get that for all $u < T^*$:
\begin{equation}\label{gronwalls2}
 \underset{N\to \infty}{\lim} \mathbb{E}\left[\sup_{t\in[0,T]} |E^N_{t \wedge u}|^{2p}\right] \le C \sup_{k \in [0,u]} \underset{N\to \infty}{\lim}\mathbb{E}\left[\mathcal{W}^{2p}_{2}(\mu_{k}^{N}, \mu_{k}^{Z}) \right] =0.
\end{equation}

\noindent Following similar steps as above for $|E^N_{t \wedge \tau}|^{2}$ and setting $\tau = T$ gives
\begin{equation}\label{propsplit4}
    \begin{split}
        &\mathbb{E}\left[\sup_{t\in[0,T]} |E^N_{t}|^{2}\right] \le C \mathbb{E}\left[\int_0^{T}\Big(\sup_{t\in[0,T]}|E^N_{t \wedge u}|^{2}+N^{-1}\sum_{j=1}^N \sup_{t\in[0,T]}|X^{j,N}_{t\wedge u}-X^{j}_{t\wedge u}|^{2}+\mathcal{W}^{2}_{2}(\mu_{u}^{N}, \mu_{u}^{Z}) \Big) \diff g(u)\right]\\
        &\le C \mathbb{E}\left[\int_0^{T}\Big(\sup_{t\in[0,T]}|E^N_{t \wedge u}|^{2}+N^{-1}\sum_{j=1}^N \sup_{t\in[0,T]}|X^{j,N}_{t\wedge u}-X^{j}_{t\wedge u}|^{2}+\mathcal{W}^{2}_{2}(\mu_{u}^{N}, \mu_{u}^{Z}) \Big) V^i_u\diff u\right]\\
        & \le C \int_0^{T} \mathbb{E}\left[ \left(\sup_{t\in[0,T]}|E^N_{t \wedge u}|^{2}+N^{-1}\sum_{j=1}^N \sup_{t\in[0,T]}|X^{j,N}_{t\wedge u}-X^{j}_{t\wedge u}|^{2} +\mathcal{W}^{2}_{2}(\mu_{u}^{N}, \mu_{u}^{Z}) \right) V^i_u \right]\diff u\\
        & \le C \int_0^{T} \mathbb{E}\left[ V^i_u\left(\sup_{t\in[0,T]}|E^N_{t \wedge u}|^{2}\right)\right] + \mathbb{E}\left[V^i_u N^{-1}\sum_{j=1}^N \sup_{t\in[0,T]}|X^{j,N}_{t\wedge u}-X^{j}_{t\wedge u}|^{2}\right] +\mathbb{E}\left[ V^i_u\mathcal{W}^{2}_{2}(\mu_{u}^{N}, \mu_{u}^{Z})\right]\diff u\\
        &\le C \int_0^{T} \mathbb{E}\left[ \sup_{t\in[0,T]}|E^N_{t \wedge u}|^{2p}\right]^{\frac{1}{p}}\mathbb{E}\left[{V^i_u}^{\frac{p}{p-1}} \right]^{\frac{p-1}{p}}+\mathbb{E}\left[N^{-1}\sum_{j=1}^N \sup_{t\in[0,T]}|X^{j,N}_{t\wedge u}-X^{j}_{t\wedge u}|^{2p} \right]^{\frac{1}{p}}  \mathbb{E}\left[{V^i_u}^{\frac{p}{p-1}} \right]^{\frac{p-1}{p}} \diff u\\ 
        &+C \int_0^{T} \mathbb{E}\left[\mathcal{W}^{2p}_{2}(\mu_u^{N}, \mu_u^{Z})\right]^{\frac{1}{p}}\mathbb{E}\left[{V^i_u}^{\frac{p}{p-1}} \right]^{\frac{p-1}{p}}\diff u\\
        &\le C \int_0^{T}\left(\mathbb{E}\left[ \sup_{t\in[0,T]}|E^N_{t \wedge u}|^{2p}\right]^{\frac{1}{p}} + \mathbb{E}\left[\mathcal{W}^{2p}_{2}(\mu_{u}^{N}, \mu_{u}^{Z})\right]^{\frac{1}{p}}\right) \mathbb{E}\left[{V^i_u}^{\frac{p}{p-1}} \right]^{\frac{p-1}{p}} \diff u\\
        &= C \sup_{t\in [0,T]}\mathbb{E}\left[{V^i_t}^{\frac{p}{p-1}} \right]^{\frac{p-1}{p}}\int_0^{T}\left(\mathbb{E}\left[ \sup_{t\in[0,T]}|E^N_{t \wedge u}|^{2p}\right]^{\frac{1}{p}} + \mathbb{E}\left[\mathcal{W}^{2p}_{2}(\mu_{u}^{N}, \mu_{u}^{Z})\right]^{\frac{1}{p}}\right)  \diff u.\\
\end{split}
\end{equation}
By Lemma \ref{momentboundvol}, we know that $\underset{t\in [0,T]}{\sup}\mathbb{E}\left[{V^i_t}^{\frac{p}{p-1}} \right]^{\frac{p-1}{p}}$ is finite. Taking the limit $N\to \infty $ on both sides and applying Fubini's theorem gives
\begin{equation}\label{propchaoslast}
\begin{split}
&\lim_{N \to \infty} \mathbb{E}\left[\sup_{t\in[0,T]} |E^N_{t}|^{2}\right]\\
&\le C \sup_{t\in[0,T]}\mathbb{E}\left[{V^i_t}^{\frac{p}{p-1}} \right]^{\frac{p-1}{p}} \int_0^{T}\lim_{N \to \infty}\mathbb{E}\left[ \sup_{t\in[0,T]}|E^N_{t \wedge u}|^{2p}\right]^{\frac{1}{p}}+ \lim_{N \to \infty}\mathbb{E}\left[\mathcal{W}^{2p}_{2}(\mu_{u}^{N}, \mu_{u}^{Z})\right]^{\frac{1}{p}}\diff u \\ 
& \le C T \sup_{t\in[0,T]}\mathbb{E}\left[{V^i_t}^{\frac{p}{p-1}} \right]^{\frac{p-1}{p}} \sup_{k \in[0,T]}\left(\lim_{N \to \infty}\mathbb{E}\left[ \sup_{t\in[0,T]}|E^N_{t \wedge k}|^{2p}\right]^{\frac{1}{p}} + \lim_{N \to \infty}\mathbb{E}\left[\mathcal{W}^{2p}_{2}(\mu_{k}^{N}, \mu_{k}^{Z})\right]^{\frac{1}{p}} \right).\\
\end{split}
\end{equation}
By \eqref{Wassbound3} and \eqref{gronwalls2}, $\underset{k \in[0,T]}{\sup} \underset{N \to \infty}{\lim}\mathbb{E}\left[\mathcal{W}^{2p}_{2}(\mu_{k}^{N}, \mu_{k}^{Z})\right]^{\frac{1}{p}} $ and $\underset{k \in[0,T]}{\sup} \underset{N \to \infty}{\lim}\mathbb{E}\left[ \underset{t \in[0,T]}{\sup}|E^N_{t \wedge k}|^{2p}\right]^{\frac{1}{p}} $are both equal to $0$, and therefore so is the RHS of equation \eqref{propchaoslast}. This concludes the proof.
\bigbreak

\end{proof}
\end{subsection}
\section{Time-discretisation scheme}\label{timediscretisation}
To simulate the particle system \eqref{particlessystem} on $[0,T]$, we consider a time--grid with $M$ uniform time-steps of width $\Delta t = T/M$, and use the classical Euler--Maruyama (EM) scheme for the state process $X$, and a full--truncation Euler (FTE) scheme for the squared volatility process $V$.\bigbreak

\noindent Let $\{t_0=0,t_1,t_2, ... ,t_M=T\}$ denote the time discretisation of $[0,T]$ so that $t_m = m\Delta t$, and for $m \in \{0,1,...,M-1\}$,
\begin{equation} \label{EM}
\begin{split}
&\hat{X}^{i,N}_{t_{m+1}} = \hat{X}^{i,N}_{t_{m}}+ \beta(t_m,(\hat{X}^{i,N}_{t_m},\hat{V}^{i,+}_{t_m}),\mu_{t_m}^{\hat{Z}^N})\Delta t+\sigma(t_m,(\hat{X}^{i,N}_{t_m},\hat{V}^{i,+}_{t_m}),\mu_{t_{m}}^{\hat{Z}^N})\Delta W^{x,i}_{t_m}, \, \hat{X}^{i,N}_{0} = X^i_0,\\
&\hat{V}^{i}_{t_{m+1}} = \hat{V}^{i}_{t_{m}} + k(\theta -\hat{V}^{i,\!+}_{t_{m}})\Delta t+ \xi\sqrt{\hat{V}^{i,\!+}_{t_{m}}} \Delta W^{v}_{t_m},\, \hat{V}^{i}_{0} = V^i_0,\\
\end{split}
\end{equation}
where $\Delta W^{\cdot,i}_{t_m}=(W^{\cdot,i}_{t_{m+1}}-W^{\cdot,i}_{t_m})\sim N(0,\Delta t)$, increments $\Delta W^{x,i}_{t_m}, \Delta W^{v,i}_{t_m}$ have correlation $\rho$, and $\hat{V}^{i,\!+}_{t_m} = \text{max}(\hat{V}^{i}_{t_m},0)$ for all particles $i \in \{1,...,N\}$.\bigbreak
\noindent It is well--established (see, e.g., \cite{KloePlat}) that the standard explicit EM scheme achieves strong convergence with an order of $1/2$ in the step--size for classical SDEs with Lipschitz continuous drift and diffusion coefficients. In the Introduction, we reviewed existing results on the strong convergence of the EM scheme for MV SDEs with H\"{o}lder--continuous coefficients. However, these results do not apply to our setting, where the diffusion coefficient includes a mean-field term and is unbounded. In this work, we establish strong convergence under Assumption \ref{A3}.\bigbreak 

\noindent We first introduce the continuous-time version of the discretised process defined in \eqref{EM}. Let $m_t := \underset{m\in\{0,...,M-1\}}{\max}\{t_m \le t\} $, $t' := \underset{t_m\in\{t_0,...,t_{M-1}\}}{\max}\{t_m \le t\}$.\bigbreak 
\noindent For $t\in[0,T]$, we define the continuous-time processes by
\begin{equation} \label{continuoustimehlsv}
\begin{split}
&\hat{X}^{i,N}_t = \hat{X}^{i,N}_{t'} + \beta(t',(\hat{X}^{i,N}_{t'},\hat{V}^{i,+}_{t'} ),\mu_{t'}^{\hat{Z}^{N}})(t-t')+ \sigma(t',(\hat{X}^{i,N}_{t'},\hat{V}^{i,+}_{t'} ),\mu_{t'}^{\hat{Z}^{N}})(W_t^{x,i}-W_{t'}^{x,i}),\\
&\tilde{V}^{i}_{t} = \tilde{V}^{i}_{t'} + k(\theta -\tilde{V}^{i,\!+}_{t'})(t-t')+ \xi\sqrt{\tilde{V}^{i,\!+}_{t'}} (W_t^{v}-W_{t'}^{v}),\\
\end{split}
\end{equation}
with $ (\hat{X}^{i,N}_0,\tilde{V}^i_0) = (X^i_0,V^i_0)$ for all particles $i \in \{1,...,N\}$.\bigbreak

\noindent Also, for convenience of notation, we define the piecewise constant process
\begin{equation}\label{piecewiseconstvol}
    \hat{V}^i_t := \tilde{V}^{i,+}_{t'}.
\end{equation}

\noindent To show the convergence of the approximated log-spot process, we use the following established results for the full truncation Euler discretised squared volatility process.

\begin{subsection}{The full truncation Euler (FTE) scheme}
\begin{lemma}
    The FTE scheme has uniformly bounded moments so that for $\tilde{V}^{i}_{t}$ as in \eqref{continuoustimehlsv}, 
    \begin{equation*}
        \mathbb{E}\left[ \sup_{t\in [0,T]}|\tilde{V}^{i}_{t}|^p \right] < \infty,\, \forall p\ge 1,\, i \in \{1,...,N\}, \, \, M\ge 1.
    \end{equation*}
\end{lemma} 
\begin{proof}
    The proof follows by Proposition 3.7 in \cite{CozMarRei}, together with applying the Burkholder--Davis--Gundy inequality.
\end{proof}
The following proposition, previously established, demonstrates the strong convergence of the discretised volatility process in $L^p$.
\begin{proposition}{(Theorem 1.1 in \cite{CozRei})}\label{convdiscretisedvol}
    Assume that $\nu >3$ and let $2 \le p < \nu-1$. Let ${V}^i_{t}$ be the solution to 
 \eqref{volparticles} and $\hat{V}^{i,N}_{t}$ the solution to 
\eqref{piecewiseconstvol}. Then there exist $M_0 \in \mathbb{N}$ and a positive constant $C>0$ such that $\forall \, i \in \{1,...,N\}, \, M > M_0$, 
    \begin{equation*}
       \sup_{t\in [0,T]} \mathbb{E}[|{V}^{i}_{t}-\hat{V}^{i}_{t}|^p]^{\frac{1}{p}}\le CM^{-\frac{1}{2}},
    \end{equation*}
    i.e., the discretised process converges strongly with order $1/2$ in $L^p$.
\end{proposition}
\end{subsection}
\begin{subsection}{Strong convergence of Euler--Maruyama (EM) scheme}

We are now well-equipped to demonstrate the strong convergence of the discretised log-spot process in $L^2$, achieving a rate of $1/2$ in time, up to a logarithmic factor. Our proof builds on ideas from Theorem 2.2 and Lemma 3.4 in \cite{BerBosDio}, as well as Proposition 4.3 in \cite{CozRei2}. While the aforementioned works also address non-Lipschitz diffusion coefficients, our setting poses an additional challenge due to the law dependence of the coefficients.

\begin{theorem}\label{hlsvEM}
    Let $X^{i,N}$ be the solution to the particle system (\ref{particlessystem}) and $\hat{X}^{i,N}$ its approximation (\ref{continuoustimehlsv}). Furthermore, suppose that $X^{i}_0 \in L^2(\mathbb{R})$ and $\nu >{\nu}^* = 2 + \sqrt{3}$. Under Assumption \ref{A3}, we have that for all $T< T^*$ and $M >M_0$,
    \begin{equation*}
        \underset{i\in \{1,..., N\}}{\sup}\mathbb{E}\left[\sup_{t\in[0,T]}|X^{i,N}_t-\hat{X}^{i,N}_t|^2 \right]^{1/2} \le C \sqrt{\log(2M)}M^{-\frac{1}{2}},
    \end{equation*}
for a positive constant $C$.
\end{theorem}

\begin{proof}
For fixed $i$, let $E^i_t := X^{i,N}_t-\hat{X}^{i,N}_t$, $E^i_0 = 0$ and $\Delta {X}^{i,N}_t := {X}^{i,N}_t-{X}^{i,N}_{t'}$. To simplify the notation, we drop the superscript $i$ from $E^i_t$ from now on.  We apply It\^{o}'s formula to $|E_{t \wedge \tau}|^{2p}$, where $\tau$ is a stopping time and  $1<p< \frac{\nu^2 - \nu}{3\nu -1}$, to get that
  \begin{equation*}
  \begin{split}
     |E_{t \wedge \tau}|^{2p} &= {2p}\int_0^{t \wedge \tau} |E_u|^{2p-1}\text{sgn}(E_u)\big(\beta(u,(X^{i,N}_u,V^i_u),\mu_u^{Z^N}) - \beta(u',(\hat{X}^{i,N}_{u'},\hat{V}^{i}_{u} ),\mu_{u'}^{\hat{Z}^{N}})\big)\diff u\\
     &+2p\int_0^{t \wedge \tau} |E_u|^{2p-1}\text{sgn}(E_u)\left(\sqrt{V^i_u}\tilde{\sigma}(u,X^{i,N}_u,\mu_u^{Z^N}) - \sqrt{\hat{V}^{i}_u}\tilde{\sigma}(u',\hat{X}^{i,N}_{u'},\mu_{u'}^{\hat{Z}^{N}})\right)\diff W^{x,i}_u\\  
     &+\frac{1}{2}2p(2p-1)\int_0^{t \wedge \tau} |E_u|^{2p-2}\left(\sqrt{V^i_u}\tilde{\sigma}(u,X^{i,N}_u,\mu_u^{Z^N}) - \sqrt{\hat{V}^{i}_u}\tilde{\sigma}(u',\hat{X}^{i,N}_{u'},\mu_{u'}^{\hat{Z}^{N}})\right)^2\diff u,\\  
  \end{split}
\end{equation*}
where $\text{sgn}(E)=1$ if $E>0$, and $\text{sgn}(E)=-1$ otherwise.\newline
Taking the supremum over time (noting that the supremum value of the integrals over time is at $t=T$) and expectations on both sides, we get 
  \begin{equation}\label{error1}
  \begin{split}
     \mathbb{E}&\left[\sup_{t\in[0,T]}|E_{t \wedge \tau}|^{2p}\right] = 2p\mathbb{E}\Bigg[\int_0^{T \wedge \tau} |E_u|^{2p-1}|\beta(u,(X^{i,N}_u,V^i_u),\mu_u^{Z^N}) - \beta(u',(\hat{X}^{i,N}_{u'},\hat{V}^{i}_{u} ),\mu_{u'}^{\hat{Z}^{N}})|\diff u\Bigg]\\
     &+2p\mathbb{E}\left[\sup_{t\in[0,T]}\int_0^{t \wedge \tau} |E_u|^{2p-1}\text{sgn}(E_u)\left(\sqrt{V^i_u}\tilde{\sigma}(u,X^{i,N}_u,\mu_u^{Z^N}) - \sqrt{\hat{V}^{i}_u}\tilde{\sigma}(u',\hat{X}^{i,N}_{u'},\mu_{u'}^{\hat{Z}^{N}})\right)\diff W^{x,i}_u\right]\\  
     &+\frac{1}{2}2p(2p-1)\mathbb{E}\left[\int_0^{T \wedge \tau} |E_u|^{2p-2} |\sqrt{V^i_u}\tilde{\sigma}(u,X^{i,N}_u,\mu_u^{Z^N}) - \sqrt{\hat{V}^{i}_u}\tilde{\sigma}(u',\hat{X}^{i,N}_{u'},\mu_{u'}^{\hat{Z}^{N}})|^2\diff u\right].\\  
  \end{split}
\end{equation}
We now bound each term above separately. Looking at the second term, to remove the stochastic integral we apply the BDG inequality to get that for $C_p>0$,
\begin{equation}\label{subst1}
\begin{split}
    &2p\mathbb{E}\left[\sup_{t\in[0,T]}\int_0^{t \wedge \tau} |E_u|^{2p-1}\text{sgn}(E_u)\left(\sqrt{V^i_u}\tilde{\sigma}(u,X^{i,N}_u,\mu_u^{Z^N}) - \sqrt{\hat{V}^{i}_u}\tilde{\sigma}(u',\hat{X}^{i,N}_{u'},\mu_{u'}^{\hat{Z}^{N}})\right)\diff W^{x,i}_u\right]\\
    &\le C_p\mathbb{E}\left[\left(\int_0^{T \wedge \tau}|E_u|^{2(2p-1)}|\sqrt{V^i_u}\tilde{\sigma}(u,X^{i,N}_u,\mu_u^{Z^N}) - \sqrt{\hat{V}^{i}_u}\tilde{\sigma}(u',\hat{X}^{i,N}_{u'},\mu_{u'}^{\hat{Z}^{N}})|^2\diff u \right)^{1/2}\right]\\
    &\le \mathbb{E}\left[\left(\sup_{t\in[0,T]}|E_{t\wedge \tau}|^{2p} C_p^2\int_0^{T \wedge \tau}|E_u|^{2p-2}|\sqrt{V^i_u}\tilde{\sigma}(u,X^{i,N}_u,\mu_u^{Z^N}) - \sqrt{\hat{V}^{i}_u}\tilde{\sigma}(u',\hat{X}^{i,N}_{u'},\mu_{u'}^{\hat{Z}^{N}})|^2\diff u \right)^{1/2}\right]\\
 &\le \frac{1}{2}\mathbb{E}\left[\sup_{t\in[0,T]}|E_{t \wedge \tau}|^{2p}\right] + \frac{1}{2}C_p^2\mathbb{E}\left[\int_0^{T \wedge \tau}|E_u|^{2p-2}|\sqrt{V^i_u}\tilde{\sigma}(u,X^{i,N}_u,\mu_u^{Z^N}) - \sqrt{\hat{V}^{i}_u}\tilde{\sigma}(u',\hat{X}^{i,N}_{u'},\mu_{u'}^{\hat{Z}^{N}})|^2\diff u \right],\\
\end{split}
\end{equation}
where the last step follows from the AM--GM inequality. We now treat the first term in \eqref{error1}. 
\begin{equation*}
  \begin{split}
     \Lambda_1 &:= \mathbb{E}\Bigg[\int_0^{T\wedge\tau} |E_u|^{2p-1}|\beta(u,(X^{i,N}_u,V^i_u),\mu_u^{Z^N}) - \beta(u',(\hat{X}^{i,N}_{u'},\hat{V}^{i}_{u} ),\mu_{u'}^{\hat{Z}^{N}})|\diff u\Bigg]\\
     &\le \mathbb{E}\Bigg[\int_0^{T\wedge\tau} |E_u|^{2p-1}V^i_u|\tilde{\beta}(u,X^{i,N}_u,\mu_u^{Z^N}) - \tilde{\beta}(u',\hat{X}^{i,N}_{u'},\mu_{u'}^{\hat{Z}^{N}})| + \tilde{\beta}_{\max}|E_u|^{2p-1}|V^i_u-\hat{V}^{i}_{u}|\diff u\Bigg]\\
     &\le \mathbb{E}\Bigg[\int_0^{T\wedge\tau}L_{\tilde{\beta}} |E_u|^{2p-1}V^i_u\left(|u-u'|^{1/2}+|X^{i,N}_u-\hat{X}^{i,N}_{u'}|+\mathcal{W}_2(\mu_u^{Z^N},\mu_{u'}^{\hat{Z}^{N}})\right)\\
     & + \quad \quad \quad \quad \quad \quad \quad \quad \quad \quad \quad \quad \quad \quad \quad \quad \quad \quad \quad \quad \quad \quad \quad \quad \quad \quad +\tilde{\beta}_{\max}|E_u|^{2p-1}|V^i_u-\hat{V}^{i}_{u}|\diff u\Bigg]\\
     &\le \mathbb{E}\Bigg[\int_0^{T\wedge\tau}L_{\tilde{\beta}} |E_u|^{2p-1}V^i_u\left(|u-u'|^{1/2}+|\Delta X^{i,N}_u|+|E_{u'}|+\mathcal{W}_2(\mu_u^{Z^N},\mu_{u'}^{\hat{Z}^{N}})\right)+\\
     & \quad \quad \quad \quad \quad \quad \quad \quad \quad \quad \quad \quad \quad \quad \quad \quad \quad \quad \quad \quad \quad \quad \quad \quad \quad \quad +\tilde{\beta}_{\max}|E_u|^{2p-1}|V^i_u-\hat{V}^{i}_{u}|\diff u\Bigg].\\
  \end{split}
\end{equation*}
Applying Young's inequality, we get the following bounds
\begin{equation}
\begin{split}
    |E_u|^{2p-1}V^i_u|u-u'|^{1/2} &\le \frac{2p-1}{2p}|E_u|^{2p}+\frac{1}{2p}{V^i_u}^{2p}|u-u'|^{p},\\
    |E_u|^{2p-1}V^i_u|\Delta {X}^{i,N}_u|&\le \frac{2p-1}{2p}|E_u|^{2p}+\frac{1}{2p}{V^i_u}^{2p}|\Delta {X}^{i,N}_u|^{2p},\\
    \tilde{\beta}_{\max}|E_u|^{2p-1}|V^{i}_u-\hat{V}^{i}_u|&\le \frac{2p-1}{2p}|E_u|^{2p}+\frac{1}{2p}\tilde{\beta}_{\max}^{2p}|V^{i}_u-\hat{V}^{i}_u|^{2p},\\
    |E_u|^{2p-1}V^i_u\mathcal{W}_2(\mu_u^{Z^N},\mu_{u'}^{\hat{Z}^{N}})&\le \frac{2p-1}{2p}|E_u|^{2p}+\frac{1}{2p}{V^i_u}^{2p}\mathcal{W}_2(\mu_u^{Z^N},\mu_{u'}^{\hat{Z}^{N}})^{2p}\\
    &\le \frac{2p-1}{2p}|E_u|^{2p}+\frac{1}{2p}{V^i_u}^{2p}\Big(\mathbb{E}\left[|X^{i,N}_u-\hat{X}^{i,N}_{u'}|^2+|V^{i}_u-\hat{V}^{i}_u|^2 \right]\Big)^{p},\\
\end{split}
\end{equation}
which follows by the definition of the Wasserstein metric. Further, we have that
\begin{equation}\label{Wassbound1}
\begin{split}
 |X^{i,N}_u-\hat{X}^{i,N}_{u'}|^2+|V^{i}_u-\hat{V}^{i}_u|^2 &\le 2|E_{u'}|^2+2|\Delta X^{i,N}_u|^2+|V^{i}_u-\hat{V}^{i}_u|^2\\
 &\le 3^{\frac{p-1}{p}}\Big(2^p|E_{u'}|^{2p}+2^p|\Delta X^{i,N}_u|^{2p}+|V^{i}_u-\hat{V}^{i}_u|^{2p}\Big)^{\frac{1}{p}},\\
\end{split}
\end{equation}
where we applied H\"{o}lder's inequality for $p>1$ in the second step. Taking the expectation on both sides and applying H\"{o}lder's inequality for $p>1$ once again, we get that
\begin{equation*}
\begin{split}
 \mathbb{E}\left[|X^{i,N}_u-\hat{X}^{i,N}_{u'}|^2+|V^{i}_u-\hat{V}^{i}_u|^2\right] & \le \mathbb{E}\left[\big(|X^{i,N}_u-\hat{X}^{i,N}_{u'}|^2+|V^{i}_u-\hat{V}^{i}_u|^2\big)^{p}\right]^{\frac{1}{p}}\\
 &\le 3^{p-1}\mathbb{E}\left[2^p|E_{u'}|^{2p}+2^p|\Delta X^{i,N}_u|^{2p}+|V^{i}_u-\hat{V}^{i}_u|^{2p}\right]^{\frac{1}{p}}.\\
\end{split}
\end{equation*}
Finally, raising both sides to the power of $p$ we get that
\begin{equation}\label{Wassbound2}
 \Big(\mathbb{E}\left[|X^{i,N}_u-\hat{X}^{i,N}_{u'}|^2+|V^{i}_u-\hat{V}^{i}_u|^2\right]\Big)^p\le 3^{p(p-1)}\mathbb{E}\left[2^p|E_{u'}|^{2p}+2^p|\Delta X^{i,N}_u|^{2p}+|V^{i}_u-\hat{V}^{i}_u|^{2p}\right].
\end{equation}
We now use Fubini's theorem and the fact that 
\begin{equation*}
    |E_{u'}| \le \sup_{t\in [0,u]}|E_{t}|, \text{\, \,} \mathbb{E}[|E_{u'}|^2]\le \mathbb{E}[\sup_{t\in [0,u]}|E_{t}|^2]\text{\, and \,} \mathbb{E}[|V^{i}_u-\hat{V}^{i}_u|^2] \le \sup_{t\in[0,u]}\mathbb{E}[|V^{i}_t-\hat{V}^{i}_t|^2],
\end{equation*} to get

\begin{equation*}
  \begin{split}
     &\Lambda_1 \le \mathbb{E}\int_0^{T\wedge\tau} \frac{2p-1}{2p}(4L_{\tilde{\beta}}+1)|E_u|^{2p}+\frac{L_{\tilde{\beta}}}{2p}V^i_u|E_{u'}|^{2p} +\frac{L_{\tilde{\beta}}}{2p}{V^i_u}^{2p}|u-u'|^{p} + \frac{L_{\tilde{\beta}}}{2p}{V^i_u}^{2p} |\Delta X^{i,N}_u|^{2p} \\
     &+ L_{\tilde{\beta}} \frac{3^{p(p-1)}}{2p}{V^i_u}^{2p} \mathbb{E}\left[2^p |E_{u'}|^{2p} + 2^{p}|\Delta X^{i,N}_u|^{2p}+|V^i_u - \hat{V}^i_u|^{2p}\right] + \frac{1}{2p}\tilde{\beta}^{2p}_{\max}|V^i_u - \hat{V}^i_u|^{2p}\diff u \\ 
     & \le C_{L_{\tilde{\beta}},p} \Bigg(T\Delta t^p \mathbb{E}[\sup_{t \in [0,T]}{V^i_t}^{2p}] + T\mathbb{E}\left[ \sup_{t\in [0,T]} {V^i_t}^{2p} \sup_{t\in [0,T]}|\Delta {X}^{i,N}_t|^{2p} \right] \\
     &+ \mathbb{E}\left[\int_0^{T\wedge\tau} (V^i_u+1)\sup_{t\in [0,u]}|E_t|^{2p}\diff u\right] +T {\tilde{\beta}_{\max}}^{2p} \sup_{t\in [0,T]} \mathbb{E}[|V^i_t - \hat{V}^i_t|^{2p}] \\
     &+T \mathbb{E}\left[ \sup_{t\in [0,T]} {V^i_t}^{2p}\right] \mathbb{E}\left[\sup_{t\in [0,T]}|\Delta {X}^{i,N}_t|^{2p} \right] +T \mathbb{E}\left[ \sup_{t\in [0,T]} {V^i_t}^{2p}\right] \mathbb{E}\left[ \sup_{t\in [0,T]}|V^{i}_t-\hat{V}^{i}_t|^{2p}\right] \\
     &+ \mathbb{E}\left[\sup_{t\in [0,T]}{V^i_t}^{2p}\right] \mathbb{E}\left[\int_0^{T \wedge \tau} \sup_{t\in [0,u]}|E_{t}|^{2p} \diff u\right] \Bigg),\\
    \end{split}
\end{equation*}
where $C_{L_{\tilde{\beta}},p}$ is a positive constant.\newline
By Lemma \ref{momentboundvol}, Proposition \ref{convdiscretisedvol}, and the Cauchy--Schwarz inequality together with Theorem 1 in \cite{FisNap} for the process $X^{i,N}$, we get that for $C>0$,
\begin{equation}\label{bounds}
\begin{split}
&\mathbb{E}\left[ \sup_{t\in [0,T]} {V^i_t}^{2p} \sup_{t\in [0,T]}|\Delta {X}^{i,N}_t|^{2p} \right] \le \mathbb{E}\left[ \sup_{t\in [0,T]} {V^i_t}^{4p}\right]^{\frac{1}{2}} \mathbb{E}\left[ \sup_{t\in [0,T]}|\Delta {X}^{i,N}_t|^{4p} \right]^{\frac{1}{2}} \le C \Big(\frac{\text{log}(2M)}{M}\Big)^p,\\
& \mathbb{E}\left[ \sup_{t\in [0,T]}|\Delta {X}^{i,N}_t|^{2p} \right] \le C \left(\frac{\text{log}(2M)}{M}\right)^p,\\
\end{split}
\end{equation}
so that 
\begin{equation}\label{subst2}
\Lambda_1 \le C\left(M^{-p}+\left(\frac{M}{\text{log}(2M)}\right)^{-p}+ \mathbb{E}\left[\int_0^{T \wedge \tau} \sup_{t\in [0,u]}|E_{t}|^{2p} \big(V^i_u +1 \big)\diff u\right]\right).
\end{equation}
Now, looking at the third term in \eqref{error1} and using similar techniques as above we get 

\begin{equation}\label{secondsplit}
\begin{split}
&\mathbb{E}\left[\int_0^{T \wedge \tau} |E_u|^{2p-2} \left|\sqrt{V^i_u}\tilde{\sigma}(u,X^{i,N}_u,\mu_u^{Z^N}) - \sqrt{\hat{V}^{i}_u}\tilde{\sigma}(u',\hat{X}^{i,N}_{u'},\mu_{u'}^{\hat{Z}^{N}})\right|^2 \diff u\right] \\
&\le 2 \mathbb{E}\Bigg[\int_0^{T \wedge \tau} |E_u|^{2p-2} |\sqrt{V^i_u}|^2 \left|\tilde{\sigma}(u,X^{i,N}_u,\mu_u^{Z^N}) - \tilde{\sigma}(u',\hat{X}^{i,N}_{u'},\mu_{u'}^{\hat{Z}^{N}})\right|^2 \\
&\quad\quad\quad\quad\quad\quad\quad\quad\quad \quad\quad\quad\quad\quad\quad\quad\quad\quad\quad + |E_u|^{2p-2} \left|\tilde{\sigma}(u',\hat{X}^{i,N}_{u'},\mu_{u'}^{\hat{Z}^{N}})\right|^2 \left|\sqrt{V^i_u}-\sqrt{\hat{V}^{i}_u}\right|^2 \diff u \Bigg] \\
&\le 2 \mathbb{E}\Bigg[\int_0^{T \wedge \tau} L_{\tilde{\sigma}}^2 |E_u|^{2p-2} V^i_u \left(|u-u'|^{1/2} + |X^{i,N}_u - \hat{X}^{i,N}_{u'}| + \mathcal{W}_2(\mu_u^{Z^N},\mu_{u'}^{\hat{Z}^{N}})\right)^2 \\
&\quad\quad\quad\quad \quad\quad\quad\quad\quad\quad\quad\quad\quad\quad\quad\quad\quad\quad\quad\quad\quad\quad\quad\quad+ |E_u|^{2p-2} |\tilde{\sigma}_{\max}|^2 \left|\sqrt{V^i_u} - \sqrt{\hat{V}^{i}_u}\right|^2 \diff u \Bigg] \\
&\le 2 \mathbb{E}\Bigg[\int_0^{T \wedge \tau} L_{\tilde{\sigma}}^2 |E_u|^{2p-2} V^i_u \left(|u-u'|^{1/2} + |E_{u'}| + |\Delta{X}^{i,N}_u| + \mathcal{W}_2(\mu_u^{Z^N},\mu_{u'}^{\hat{Z}^{N}})\right)^2 \\
&\quad\quad\quad\quad\quad\quad\quad\quad \quad\quad\quad\quad \quad\quad\quad\quad\quad\quad\quad\quad \quad\quad\quad\quad+ |E_u|^{2p-2} |\tilde{\sigma}_{\max}|^2 \left|\sqrt{V^i_u} - \sqrt{\hat{V}^{i}_u}\right|^2 \diff u \Bigg] \\
&\le C \mathbb{E}\Bigg[\int_0^{T \wedge \tau} L_{\tilde{\sigma}}^2 |E_u|^{2p-2} V^i_u \left(\Delta t + |E_{u'}|^2 + |\Delta{X}^{i,N}_u|^2 + \mathcal{W}^2_2(\mu_u^{Z^N},\mu_{u'}^{\hat{Z}^{N}})\right) \\
&\quad\quad\quad\quad \quad\quad\quad\quad \quad\quad\quad\quad \quad\quad\quad\quad\quad\quad\quad\quad\quad\quad\quad\quad + |E_u|^{2p-2} |\tilde{\sigma}_{\max}|^2 \left|\sqrt{V^i_u} - \sqrt{\hat{V}^{i}_u}\right|^2 \diff u \Bigg]. \\
\end{split}
\end{equation}

\noindent We now treat each term in the equation above separately using similar techniques as above. Specifically, we apply Young's inequality to get 
\begin{equation}
\begin{split}
    |E_u|^{2p-2}V^i_u \Delta t &\le \frac{2p-2}{2p}|E_u|^{2p}+\frac{1}{p}{V^i}^p_u \Delta t^p ,\\
    |E_u|^{2p-2}V^i_u |\Delta {X}^{i,N}_u|^2&\le \frac{2p-2}{2p}|E_u|^{2p}+\frac{1}{p}{V^i}^p_u|\Delta {X}^{i,N}_u|^{2p},\\
    |E_u|^{2p-2}V^i_u\mathcal{W}^2_2(\mu_u^{Z^N},\mu_{u'}^{\hat{Z}^{N}})&\le \frac{2p-2}{2p}|E_u|^{2p}+\frac{1}{p}{V^i}^p_u\mathcal{W}^{2p}_{p}(\mu_u^{Z^N},\mu_{u'}^{\hat{Z}^{N}})\\
    &\le \frac{p-1}{p}|E_u|^{2p}+\frac{1}{p}{V^i}^p_u\Big(\mathbb{E}\left[|X^{i,N}_u-\hat{X}^{i,N}_{u'}|^2+|V^{i}_u-\hat{V}^{i}_u|^2 \right]\Big)^{p}\\
    & \le \frac{p-1}{p}|E_u|^{2p} + \frac{3^{p(p-1)}}{p}{V^i}^p_u \mathbb{E}\left[2^p|E_{u'}|^{2p}+2^p|\Delta X^{i,N}_u|^{2p}+|V^{i}_u-\hat{V}^{i}_u|^{2p}\right],\\
\end{split}
\end{equation}
where the last step follows by using bound \eqref{Wassbound2} derived above. Substituting the above bounds back in \eqref{secondsplit} we get 

\begin{equation}\label{fourthsplit}
\begin{split}
&\mathbb{E}\left[\int_0^{T \wedge \tau} |E_u|^{2p-2} \left|\sqrt{V^i_u}\tilde{\sigma}(u,X^{i,N}_u,\mu_u^{Z^N}) - \sqrt{\hat{V}^{i}_u}\tilde{\sigma}(u',\hat{X}^{i,N}_{u'},\mu_{u'}^{\hat{Z}^{N}})\right|^2 \diff u\right]\\
& \le C \mathbb{E}\Bigg[\int_0^{T \wedge \tau}  |E_u|^{2p} + {V^i_u}^p \Delta t^p + V^i_u |E_{u'}|^2 |E_u|^{2p-2} + {V^i_u}^p |\Delta{X}^{i,N}_u|^{2p} \\
& \qquad + {V^i_u}^p \mathbb{E}\Big[2^p|E_{u'}|^{2p} + 2^p|\Delta X^{i,N}_u|^{2p} + |V^i_u - \hat{V}^i_u|^{2p}\Big] + |\tilde{\sigma}_{u'}|^{2p} \left|\sqrt{V^i_u} - \sqrt{\hat{V}^i_u}\right|^{2p} \diff u\Bigg]\\
& \le C \Bigg(T\Delta t^p \mathbb{E}\big[\sup_{t \in [0,T]} {V^i_t}^p\big] 
+ T\mathbb{E}\big[\sup_{t \in [0,T]} {V^i_t}^p \sup_{t \in [0,T]} |\Delta {X}^{i,N}_t|^{2p}\big] \\
& \qquad + \mathbb{E}\Bigg[\int_0^{T \wedge \tau} (V^i_u+1) \sup_{t \in [0,u]} |E_t|^{2p} \diff u\Bigg] 
+ \mathbb{E}\big[\sup_{t \in [0,T]} {V^i_t}^p\big] \mathbb{E}\Bigg[\int_0^{T \wedge \tau} \sup_{t \in [0,u]} |E_t|^{2p} \diff u\Bigg] \\
& \qquad + T \mathbb{E}\big[\sup_{t \in [0,T]} {V^i_t}^p\big] \mathbb{E}\big[\sup_{t \in [0,T]} |\Delta {X}^{i,N}_t|^{2p}\big] 
+ T \mathbb{E}\big[\sup_{t \in [0,T]} {V^i_t}^p\big] \mathbb{E}\big[\sup_{t \in [0,T]} |V^i_t - \hat{V}^i_t|^{2p}\big] \\
& \qquad + T \tilde{\sigma}_{\text{max}}^{2p} \sup_{t \in [0,T]} \mathbb{E}\big[|\sqrt{V^i_t} - \sqrt{\hat{V}^i_t}|^{2p}\big] \Bigg),
\end{split}
\end{equation}
which follows from Fubini's theorem. Using similar arguments as above, see \eqref{bounds}, we get that terms $\mathbb{E}\left[ \underset{t\in [0,T]}{\sup} {V^i_t}^p \underset{t\in [0,T]}{\sup}|\Delta {X}^{i,N}_t|^{2p} \right]$ and $\mathbb{E}\left[ \underset{t\in [0,T]}{\sup}|\Delta {X}^{i,N}_t|^{2p} \right] $ are bounded above by $C \Big(\frac{\text{log}(2M)}{M}\Big)^p$ for $C$ a positive constant.\newline
Also, since $\nu \ge 1$, the squared volatility process has strictly positive paths almost surely, so that we can write 
\begin{equation}
   |\sqrt{V^i_t}-\sqrt{\hat{V}^{i}_t}|^2 \le  {V^i_t}^{-1}|V^i_t-\hat{V}^{i}_t|^2,
\end{equation}
so that by H\"{o}lder's inequality, Lemma \ref{momentboundvol}, and Proposition \ref{convdiscretisedvol}, we have that
\begin{equation}
   \sup_{t \in [0,T]}\mathbb{E}\left[|\sqrt{V^i_t}-\sqrt{\hat{V}^{i}_t}|^{2p} \right]\le \sup_{t \in [0,T]}\mathbb{E}\left[|V^i_t|^{-pr}\right]^{\frac{1}{r}} \sup_{t \in [0,T]}\mathbb{E}\left[|V^i_t-\hat{V}^{i}_t|^{\frac{2pr}{r-1}}\right]^{\frac{r-1}{r}} \le CM^{-p},
\end{equation}
for $r>1$ such that $-pr>- \nu$ and $2 \le \frac{2pr}{r-1} < \nu -1$ that after some algebra are combined to
\begin{equation*}
    \frac{\nu}{\nu-p}<\frac{r}{r-1}<\frac{\nu-1}{2p}.
\end{equation*}
For the above condition to be satisfied, we need to choose $p$ such that $p< \frac{\nu^2 - \nu}{3\nu -1}$. Recall that $p>1$ so that we need $\nu > 2+\sqrt{3}$. Substituting everything back in equation \eqref{fourthsplit}, we get
\begin{equation}\label{subst3}
\begin{split}
 \mathbb{E}&\left[\int_0^{T \wedge \tau} |E_u|^{2p-2} |\sqrt{V^i_u}\tilde{\sigma}(u,X^{i,N}_u,\mu_u^{Z^N}) - \sqrt{\hat{V}^{i}_u}\tilde{\sigma}(u',\hat{X}^{i,N}_{u'},\mu_{u'}^{\hat{Z}^{N}})|^2\diff u\right] \\
 & \le C\left(M^{-p}+\Big(\frac{M}{\text{log}(2M)}\Big)^{-p}+ \mathbb{E}\left[\int_0^{T \wedge \tau} \sup_{t\in [0,u]}|E_{t}|^{2p} \big(V^i_u +1 \big)\diff u\right]\right).
\end{split}
\end{equation}
Substituting bounds \eqref{subst1}, \eqref{subst2}, and \eqref{subst3} back in \eqref{error1} and rearranging, we get that for a positive constant $C$ that only depends on $p, T, \tilde{\beta}_{\max}$ and $\tilde{\sigma}_{\max}$,
\begin{equation}\label{subst4}
    \mathbb{E}\left[\sup_{t\in[0,T]}|E_{t \wedge \tau}|^{2p}\right] \le  C\left(M^{-p}+\Big(\frac{M}{\text{log}(2M)}\Big)^{-p}+ \mathbb{E}\left[\int_0^{T \wedge \tau} \sup_{t\in [0,u]}|E_{t}|^{2p} \big(V^i_u +1 \big)\diff u\right]\right).
\end{equation}
The following steps use ideas from \cite{CozRei2} and \cite{BerBosDio}. First, we consider the strictly increasing stochastic process
\begin{equation}
  g^i(t) := \int_0^t \Big(V^i_u+1\Big) \diff u,
\end{equation}
so that \eqref{subst4} can be written as 
\begin{equation}\label{subst5}
    \mathbb{E}\left[\sup_{t\in[0,T]}|E_{t \wedge \tau}|^{2p}\right] \le  C\left(M^{-p}+\Big(\frac{M}{\text{log}(2M)}\Big)^{-p}+ \mathbb{E}\left[\int_0^{T \wedge \tau} \sup_{t\in [0,T]}|E_{t \wedge u}|^{2p} \diff g^i(u)\right]\right).
\end{equation}
Let ${\tau}_{\omega}$ be a stopping time defined by ${\tau}^i_{\omega}:= \text{inf}\{t\ge 0 \,| \,g^i(t)\ge \omega\}$ for $\omega \ge 0$ with ${\tau}_{0}=0.$\newline
Notice that $g^i(\tau^i_{\omega})=\omega, \, g^i(T\wedge \tau^i_{\omega}) = g^i(T) \wedge g^i(\tau^i_{\omega}) = g^i(T) \wedge \omega ,$ and $\tau^i_{\omega}$ is finite since 
\begin{equation*}
    \tau^i_{\omega} = \omega - \int_0^{\tau^i_{\omega}}V^i_u \diff u \le \omega.
\end{equation*}
We now set $\tau = \tau^i_{\omega}$ for fixed $\omega >0$ in \eqref{subst5} and to simplify the notation, in what follows, we drop the superscript $i$ from $\tau^i_{\cdot}$ and $g^i(\cdot)$. Also, consider the stochastic time-change $s=g(u)$, so that $u = \tau_s$. Applying Lebesgue's change-of-time integral formula we get
\begin{equation*}
\mathbb{E}\left[\int_0^{T \wedge \tau_{\omega}} \sup_{t\in [0,T]}|E_{t \wedge u}|^{2p} \diff g(u)\right] = \mathbb{E}\left[\int_0^{g(T) \wedge \omega} \sup_{t\in [0,T]}|E_{t \wedge \tau_s}|^{2p} \diff s\right] \le \int_0^{\omega} \mathbb{E}\left[ \sup_{t\in [0,T]} |E_{t \wedge \tau_s}|^{2p} \right] \diff s,
\end{equation*}
so that 
\begin{equation*}
\mathbb{E}\left[\sup_{t\in[0,T]}|E_{t \wedge \tau_\omega}|^{2p}\right] \le C\left(\frac{M}{\text{log}(2M)} \right)^{-p}+C\int_0^{\omega} \mathbb{E}\left[ \sup_{t\in [0,T]} |E_{t \wedge \tau_s}|^{2p} \right] \diff s,
\end{equation*}
and by applying Gr\"{o}nwall's inequality we get that for all $\omega>0$ and $p> 1$,
\begin{equation}
\mathbb{E}\left[\sup_{t\in[0,T]}|E_{t \wedge \tau_\omega}|^{2p}\right] \le C e^{C\omega} \left(\frac{M}{\text{log}(2M)} \right)^{-p}. 
\end{equation}
We now take $\tau = T$ in \eqref{subst5} and use similar techniques as above to find
\begin{equation*}
\mathbb{E}\left[\sup_{t\in[0,T]}|E_{t}|^{2}\right] \le C\left(\frac{M}{\text{log}(2M)} \right)^{-1}+C\mathbb{E}\left[ \int_0^{T} \sup_{t\in [0,T]} |E_{t \wedge u}|^{2}  \diff g(u)\right].
\end{equation*}
Using a stochastic time-change as before, together with Fubini's theorem and H\"{o}lder's inequality one gets
\begin{equation*}
\begin{split}
\mathbb{E}\left[ \int_0^{T} \sup_{t\in [0,T]} |E_{t \wedge u}|^{2}  \diff g(u)\right] &\le \int_0^\infty \mathbb{E}\left[  \sup_{t\in [0,T]} |E_{t \wedge \tau_s}|^{2p} \right]^{\frac{1}{p}}   \mathbb{E}\left[ \mathbbm{1}_{s \le g(T)} \right]^{1-\frac{1}{p}} \diff s\\
& \le C \int_0^\infty e^{\frac{Cs}{p}} \left(\frac{M}{\text{log}(2M)} \right)^{-1}   \mathbb{P}\Big( s \le g(T)\Big)^{1-\frac{1}{p}} \diff s.
\end{split}
\end{equation*}
We therefore deduce that
\begin{equation}\label{FinBound} \mathbb{E}\left[\sup_{t\in[0,T]}|E_{t}|^{2}\right] \le C \frac{\text{log}(2M)}{M} \left(1+\int_0^\infty e^{\frac{Cs}{p}} \mathbb{P}\Big( s \le g(T)\Big)^{1-\frac{1}{p}} \diff s \right).
\end{equation}
By Markov's inequality, we have that for $\lambda>\frac{C}{p-1}$,
\begin{equation}\label{MarkIneq}
 \mathbb{P}\Big( s \le g(T)\Big) \le e^{-\lambda s} \mathbb{E}\left[ e^{\lambda g(T)} \right] = e^{-\lambda s} \mathbb{E}\left[ e^{\lambda \int_0^T (V^i_t+1) \diff t} \right].
\end{equation}
Substituting this bound in \eqref{FinBound}, we get that 
\begin{equation}\label{FinBound2} \mathbb{E}\left[\sup_{t\in[0,T]}|E_{t}|^{2}\right] \le C \frac{\text{log}(2M)}{M} \left(1 + \int_0^\infty e^{\frac{Cs}{p}-\lambda s ( 1-\frac{1}{p})} \mathbb{E}\left[ e^{\lambda \int_0^T (V^i_t+1) \diff t} \right]^{1-\frac{1}{p}} \diff s \right).
\end{equation}
By Lemma \ref{ExpIntVol}, we know that for $T<T^*$, the expectation on the right--hand side of equation \eqref{FinBound2} is finite. Also, since $\lambda>\frac{C}{p-1}$,
\begin{equation*}
 \int_0^{\infty}\text{exp}\Big(\frac{Cs}{p}-\lambda s\big(1-\frac{1}{p}\big)\Big) \diff s =  \int_0^{\infty}\text{exp}\Big(-\frac{s}{p}\Big(\lambda(p-1)-C\Big)\Big) \diff s = \frac{p}{\big(\lambda (p-1)-C\big)}.
\end{equation*}
The result follows by substituting the above bound in \eqref{FinBound2}.
\end{proof}

\begin{remark}
The condition on the Feller ratio in Theorem \ref{hlsvEM} may appear restrictive. An alternative approach to ensure positivity of the variance process is to employ an implicit discretisation scheme as proposed in \cite{BriAlf}, which leads to the backward Euler--Maruyama scheme analysed in \cite{Alf} for the CIR process. Retracing the proof of Theorem \ref{hlsvEM} using this implicit scheme instead of the FTE scheme for the variance component shows that we can relax the Feller condition to $\nu > 2$ in that case. For a broader discussion of discretisation schemes for the CIR process, we also refer the reader to \cite{DeeDel} and \cite{ADiop}. We choose to work with the FTE scheme for a representative example here, as it remains one of the most widely used methods in practice due to its favourable properties and low bias, as detailed in \cite{LorKoe}.
\end{remark}
\end{subsection}

\section{Proof of well-posedness of the regularised equations}\label{WellposednessSection}
In this section, we prove Theorem \ref{wellposedness}, the well-posedness of the regularised calibrated dynamics \eqref{independentparticles3}. Recall that this is a McKean--Vlasov SDE with an unbounded $1/2$-H\"{o}lder continuous diffusion coefficient. The primary challenge in this task arises from the non-Lipschitz continuity and the measure dependence of the diffusion coefficient. We first recall the definition of a strong solution for MV SDE \eqref{independentparticles3} on a given probability space $(\Omega, \mathcal{F}, \mathbb{P})$ and with respect to the fixed Brownian motion $W$ and initial condition $\xi$, as stated in, e.g., \cite{CarDel}.

\begin{definition}\label{strongsolndefn}
Process $(X_t)_{t \in [0,T]}$ is a unique strong solution of equation \eqref{independentparticles3} if:
\begin{enumerate}[label=(\roman*)]
\item it is $\mathbb{F}$-adapted and continuous,
\item $\mathbb{P}[X_0 = \xi] = 1$, 
\item $\int_0^T \left(|\beta(t, (X_t, V_t), \mu^Z_t)| + \left\| \sigma(t, (X_t, V_t), \mu^Z_t) \right\|^2\right) \diff t < \infty, \quad \mathbb{P}$--a.s.,
\item for $t \in [0,T]$, 
\begin{equation*}
X(t) = X(0)+ \int_0^t \beta(s, (X_s, V_s), \mu^Z_s) \diff s + \int_0^t \sigma(s, (X_s, V_s), \mu^Z_s) \diff W_s, \quad \mathbb{P}\text{--a.s.}, 
\end{equation*}
\item for another solution $(\tilde{X}_t)_{t \in [0,T]}$ with $\tilde{X}_0 = \xi$, $\mathbb{P}\{{X}_t = \tilde{X}_t \, \, \, \forall \, \, \,  t \in [0,T] \} = 1$.
\end{enumerate}
\end{definition}



The result of Theorem \ref{wellposedness} follows by Lemma \ref{aprioriestimates}, Proposition \ref{pathwiseuniqueness}, and Proposition \ref{existencetheorem} below.

\begin{lemma}[Estimate for a solution]\label{aprioriestimates}
Let Assumptions \ref{A3} hold. Furthermore, assume that $\nu \ge 1$ and $\mathbb{E}[|X_0|^4] < \infty$. Then for $t< T^*$ for $T^*$ as in $\eqref{Tbound}$ and $(X_t)_{t \in [0,T]}$ a solution to \eqref{independentparticles3}, we have that for $n \in \{1,2,3, 4\}$,
\[\mathbb{E}\left[\underset{s \in [0,t]}{\sup}|X_s|^n\right]< \infty.\]
\end{lemma}

\begin{proof}
Let $2 \le n \le 4$, since the case $n=1$ follows from similar techniques.
Consider $m$ such that $n < m \le 4$ and let $\tau >0$ be a stopping time.
\begin{equation*}
\begin{aligned}
&|X_{s \wedge \tau}|^m \le |X_0|^m + m \int_0^{s \wedge \tau} \langle (X_u)^{m-1}, \beta(u, (X_u, V_u), \mu^Z_u) \rangle \, \diff u \\
&+ m \int_0^{s \wedge \tau} \langle (X_u)^{m-1}, \sigma(u, (X_u, V_u), \mu^Z_u)  \diff W_u \rangle + \frac{m(m-1)}{2} \int_0^{s \wedge \tau} |X_u|^{m-2}|\sigma(u, (X_u, V_u), \mu^Z_u)|^2 \diff u.
\end{aligned}
\end{equation*}
By taking the supremum over time and expectations on both sides, we get

\begin{equation}
\begin{aligned}
&\mathbb{E}\left[\underset{s \in [0,t]}{\sup}|X_{s \wedge \tau}|^m \right] \le \mathbb{E}\left[ |X_0|^m \right] + m \mathbb{E}\left[\int_0^{t \wedge \tau} |X_u|^{m-1}  |\beta(u, (X_u, V_u), \mu^Z_u)|\, \diff u \right] \\
& + m \mathbb{E}\left[\underset{s \in [0,t]}{\sup} \int_0^{s \wedge \tau} \langle (X_u)^{m-1}, \sigma(u, (X_u, V_u), \mu^Z_u)  \diff W_u \rangle \right]\\
& + \frac{m(m-1)}{2} \mathbb{E} \left[\int_0^{t \wedge \tau} |X_u|^{m-2}|\sigma(u, (X_u, V_u), \mu^Z_u)|^2 \diff u \right] \\
&\le \mathbb{E}\left[ |X_0|^m \right] + m{\tilde{\beta}}_{\max}\mathbb{E}\left[\int_0^{t \wedge \tau} |X_u|^{m-1} V_u \, \diff u \right]  + m \beta_0 \mathbb{E}\left[ \left( \int_0^{t \wedge \tau} |X_u|^{2(m-1)} |\sigma(u,(X_u, V_u), \mu^Z_u)|^2  \diff u \right)^{1/2} \right]\\
& + \frac{m(m-1)}{2} {\tilde{\sigma}^2}_{\max}\mathbb{E} \left[\int_0^{t \wedge \tau} |X_u|^{m-2} V_u \diff u \right],
\end{aligned}
\end{equation}
where $\beta_0$ is a constant coming from the application of the Burkholder--Davis--Gundy inequality to the third term of the RHS above.
\begin{equation}\label{split1}
\begin{aligned}
&\mathbb{E}\left[\underset{s \in [0,t]}{\sup}|X_{s \wedge \tau}|^m\right] \le  \mathbb{E}\left[ |X_0|^m \right] + m {\tilde{\beta}}_{\max}\mathbb{E}\left[\int_0^{t \wedge \tau} |X_u|^{m-1} V_u \, \diff u \right]\\
&+ m \beta_0 {\tilde{\sigma}}_{\max} \mathbb{E}\left[ \left( \int_0^{t \wedge \tau}  |X_u|^{2(m-1)} V_u \diff u \right)^{1/2} \right] + \frac{m(m-1)}{2} {\tilde{\sigma}^2}_{\max} \mathbb{E} \left[\int_0^{t \wedge \tau} |X_u|^{m-2} V_u \diff u \right].
\end{aligned}
\end{equation}
Notice that applying Young's inequality to the third term of the RHS above gives
\begin{equation*}
\begin{aligned}
& m \beta_0{\tilde{\sigma}}_{\max} \mathbb{E}\left[ \left( \int_0^{t \wedge \tau}  |X_u|^{2(m-1)} V_u \diff u \right)^{1/2} \right] \le  m \beta_0 {\tilde{\sigma}}_{\max}\mathbb{E}\left[ \underset{s \in [0,t]}{\sup} |X_{s\wedge \tau}|^{m/2}\left( \int_0^{t \wedge \tau}  |X_u|^{m-2} V_u \diff u \right)^{1/2} \right]\\
& \le \frac{1}{2} \mathbb{E}\left[ \underset{s \in [0,t]}{\sup} |X_{s\wedge \tau}|^{m} \right] + \frac{(m \beta_0 {\tilde{\sigma}}_{\max})^2}{2}\mathbb{E}\left[  \int_0^{t \wedge \tau}  |X_u|^{m-2} V_u \diff u  \right].\\
\end{aligned}
\end{equation*}
Substituting back in equation \eqref{split1} and rearranging gives 
\begin{equation*}
\begin{aligned}
&\mathbb{E}\left[\underset{s \in [0,t]}{\sup}|X_{s \wedge \tau}|^m\right] \le  2\mathbb{E}\left[ |X_0|^m \right] + 2m {\tilde{\beta}}_{\max}\mathbb{E}\left[\int_0^{t \wedge \tau} |X_u|^{m-1} V_u \, \diff u \right]\\
&+ \left(m(m-1){\tilde{\sigma}^2}_{\max}+(m \beta_0 {\tilde{\sigma}}_{\max})^2\right)  \mathbb{E} \left[\int_0^{t \wedge \tau} |X_u|^{m-2} V_u \diff u \right]\\
& \le 2\mathbb{E}\left[ |X_0|^m \right] + 2 \gamma(m)  \mathbb{E} \left[\int_0^{t \wedge \tau} \left(|X_u|^{m-2}+|X_u|^{m-1}\right) V_u \diff u \right],
\end{aligned}
\end{equation*}
where $ \gamma(m) :=\max \{2m{\tilde{\beta}}_{\max}, \left(m(m-1){\tilde{\sigma}^2}_{\max}+(m \beta_0 {\tilde{\sigma}}_{\max})^2\right) \}$. Also, note that for all $m > 2$, $|X_u|^{m-2}+|X_u|^{m-1} \le 1+2 |X_u|^{m-1} \le 3+ 2|X_u|^{m} $, so that
\begin{equation*}
\begin{aligned}
&\mathbb{E}\left[\underset{s \in [0,t]}{\sup}|X_{s \wedge \tau}|^m\right] \le 2\mathbb{E}\left[ |X_0|^m \right] + 2 \gamma(m)  \mathbb{E} \left[\int_0^{t \wedge \tau} \left(3+2|X_u|^{m}\right) V_u \diff u \right]\\
& \le 2\mathbb{E}\left[ |X_0|^m \right] + 6\gamma(m) t \mathbb{E}\left[\underset{s \in [0,t]}{\sup} V_s \right] + 4\gamma(m)  \mathbb{E} \left[\int_0^{t \wedge \tau} |X_u|^{m} V_u \diff u \right].\\
\end{aligned}
\end{equation*}
Following a similar reasoning as in the proofs above, we define the strictly increasing stochastic process
\begin{equation*}
g(t) := \int_0^t V_u \diff u.
\end{equation*}
We then have that
\begin{equation*}
\begin{aligned}
&\mathbb{E}\left[\underset{s \in [0,t]}{\sup}|X_{s \wedge \tau}|^m\right] \le 2\mathbb{E}\left[ |X_0|^m \right] + 6\gamma(m) t \mathbb{E}\left[\underset{s \in [0,t]}{\sup} V_s \right] + 4\gamma(m)  \mathbb{E} \left[\int_0^{t \wedge \tau} \underset{s \in [0,u]}{\sup} |X_s|^{m} \diff g(u) \right]\\
&\le 2\mathbb{E}\left[ |X_0|^m \right] + 6\gamma(m) t \mathbb{E}\left[\underset{s \in [0,t]}{\sup} V_s \right] + 4\gamma(m)  \mathbb{E} \left[\int_0^{t \wedge \tau} \underset{s \in [0,t]}{\sup} |X_{s \wedge u}|^{m} \diff g(u) \right].\\
\end{aligned}
\end{equation*}
For $\omega \ge 0$, let $\tau_{\omega}:=\{t \ge 0 \, | \, g(t)\ge \omega\}$. Notice that $g(t \wedge \tau_{\omega}) =  g(t) \wedge \omega$. In what follows, we fix $\omega >0$ and set $\tau = \tau_{\omega }$.\bigbreak

\noindent We consider a stochastic time change $r = g(u)$ so that $u = \tau _{r}$.

\begin{equation*}
\begin{aligned}
&\mathbb{E}\left[\underset{s \in [0,t]}{\sup}|X_{s \wedge \tau_\omega}|^m\right] \le 2\mathbb{E}\left[ |X_0|^m \right] + 6\gamma(m) t \mathbb{E}\left[\underset{s \in [0,t]}{\sup} V_s \right] + 4\gamma(m)  \mathbb{E} \left[\int_0^{\omega \wedge g(t)} \underset{s \in [0,t]}{\sup} |X_{s \wedge \tau_r}|^{m} \diff r \right]\\
& \le 2\mathbb{E}\left[ |X_0|^m \right] + 6\gamma(m) t \mathbb{E}\left[\underset{s \in [0,t]}{\sup} V_s \right] + 4\gamma(m)  \int_0^{\omega}\mathbb{E} \left[ \underset{s \in [0,t]}{\sup} |X_{s \wedge \tau_r}|^{m} \right] \diff r.\\
\end{aligned}
\end{equation*}
Applying Gr\"{o}nwall's inequality gives
\begin{equation}\label{Xbound}
\mathbb{E}\left[\underset{s \in [0,t]}{\sup}|X_{s \wedge \tau_\omega}|^m\right] \le T_1(t) e^{4 \gamma(m) \omega},
\end{equation}
for $T_1 (t) = 2\mathbb{E}\left[ |X_0|^m \right] + 6\gamma(m) t \mathbb{E}\left[\underset{s \in [0,t]}{\sup} V_s \right] < \infty $ by Lemma \ref{momentboundvol}.\newline

\noindent We now repeat the same process as above for $n<m$ and set $\tau = t$ to get:

\begin{equation*}
\begin{aligned}
&\mathbb{E}\left[\underset{s \in [0,t]}{\sup}|X_{s}|^n \right] \le \mathbb{E}\left[ |X_0|^n \right] + n \tilde{\beta}_{\max}\mathbb{E}\left[\int_0^{t} |X_u|^{n-1} V_u \, \diff u \right] \\
&+ n \beta_0 \mathbb{E}\left[ \left( \int_0^{t } |X_u|^{2(n-1)} |\sigma(u, X_u, V_u, \mu^Z_u)|^2  \diff u \right)^{1/2} \right] + \frac{n(n-1)}{2} \tilde{\sigma}^2_{\max}\mathbb{E} \left[\int_0^{t } |X_u|^{n-2} V_u \diff u \right],\\
& \le 2\mathbb{E}\left[ |X_0|^n \right] + 6\gamma(n) t \mathbb{E}\left[\underset{s \in [0,t]}{\sup} V_s \right] + 4\gamma(n)  \mathbb{E} \left[\int_0^{t} \underset{s \in [0,t]}{\sup} |X_{s \wedge u}|^{n} \diff g(u) \right].\\
\end{aligned}
\end{equation*}
Regarding the last term above, after applying a stochastic time-change as before followed by Fubini's theorem and H\"{o}lder's inequality we get
\begin{equation*}
\begin{aligned}
&\mathbb{E} \left[\int_0^{t} \underset{s \in [0,t]}{\sup} |X_{s \wedge u}|^{n} \diff g(u) \right] \le \int_0^{\infty} \mathbb{E} \left[\left(\underset{s \in [0,t]}{\sup} |X_{s \wedge \tau_r}|^{n}\right) \mathbbm{1}_{r \le g(t)}\right] \diff r \\
&  \le \int_0^{\infty} \mathbb{E} \left[\left(\underset{s \in [0,t]}{\sup} |X_{s \wedge \tau_r}|^{n}\right)^{\frac{m}{n}}\right]^{\frac{n}{m}} \mathbb{E}\left[\mathbbm{1}_{r \le g(t)}\right]^{\frac{m-n}{m}}\diff r \le  \int_0^{\infty} \mathbb{E} \left[\underset{s \in [0,t]}{\sup} |X_{s \wedge \tau_r}|^{m}\right]^{\frac{n}{m}}\mathbb{P}\Big(r \le g(t)\Big)^{\frac{m-n}{m}}\diff r \\
& \le \int_0^{\infty} \left(T_1(t) e^{4 \gamma(m) r}\right)^{\frac{n}{m}}\Big( 
e^{-\lambda r} \mathbb{E}\left[ e^{\lambda \int_0^t V_s \diff s} \right] \Big)^{\frac{m-n}{m}}\diff r,\\
\end{aligned}
\end{equation*}
where we applied bound \eqref{Xbound} and Markov's inequality with $\lambda >\frac{4n\gamma(m)}{m-n}$ in the last line. By Lemma \ref{ExpIntVol}, $\mathbb{E}\left[ e^{\lambda \int_0^t V_s \diff s} \right]$ is finite for $t<T^*$, with $T^*$ as in \eqref{Tbound}.
Also, for $\lambda > \frac{4n\gamma(m)}{m-n}$,
\[ \int_0^{\infty} e^{-r\left(\frac{\lambda (m-n)-4n\gamma(m)}{m}\right)}\diff r = \left[\frac{-m}{\lambda (m-n)-4n\gamma(m)}e^{-r\left(\frac{\lambda (m-n)-4n\gamma(m)}{m}\right)}\right]_0^{\infty} = \frac{m}{\lambda (m-n)-4n\gamma(m)}. \] 
Let $T_2(t) :=  2 \mathbb{E}\left[|X_0|^n \right] + 6\gamma(n) t \mathbb{E}\left[\underset{s \in [0,t]}{\sup} V_s \right] + 4\gamma(n)T_1(t)^{\frac{n}{m}}\frac{m}{\lambda (m-n)-4n\gamma(m)}\mathbb{E}\left[ e^{\lambda \int_0^t V_s \diff s} \right]^{\frac{m-n}{m}}.$\newline
We therefore conclude that 
\begin{equation*}\label{i}
\mathbb{E}\left[\underset{s \in [0,t]}{\sup}|X_{s}|^n \right] \le T_2(t) < \infty.
\end{equation*}
\newline
\end{proof}

\begin{proposition}[Pathwise uniqueness]\label{pathwiseuniqueness}
Let Assumptions \ref{A3} hold together with $\nu \ge 1$ and $\mathbb{E}[|X_0|^4] < \infty$. Then the solution to \eqref{independentparticles3} is pathwise unique until time $T^*$, for $T^*$ as in $\eqref{Tbound}$.
\end{proposition}
\begin{proof}

Let $(X^{(1)}_t)_{t \in [0,T]}$ and $(X^{(2)}_t)_{t \in [0,T]}$ two solutions to \eqref{independentparticles3} with respect to the same initial solution $X_0$ and Brownian motion $W$. We aim to show that $\mathbb{E}\left[\underset{t \in [0,T]}{\sup}|X^{(1)}_{t}- X^{(2)}_{t}|^2\right] = 0 $.\newline

\noindent Let $\tau >0$ be a stopping time. Applying It\^{o}'s formula to $|X^{(1)}_{t \wedge \tau}- X^{(2)}_{t \wedge \tau}|^2$, and following similar arguments as in the proofs of Theorem \ref{PropChaos} and Theorem \ref{hlsvEM}, we get

\begin{equation*}
\begin{aligned}
\mathbb{E}\left[\underset{t \in [0,T]}{\sup}|X^{(1)}_{t \wedge \tau}- X^{(2)}_{t \wedge \tau}|^2\right] &\le C \mathbb{E}\left[\int_0^{T \wedge \tau} \left( |X^{(1)}_{r}- X^{(2)}_{r}|^2 + \mathcal{W}_2^2(\mu_{(1)}(r), \mu_{(2)}(r))\right) \diff g(r) \right]\\
& \le C\mathbb{E}\left[ \int_0^{T \wedge \tau}  \underset{t \in [0,T]}{\sup}|X^{(1)}_{t \wedge r}- X^{(2)}_{t \wedge r}|^2 + \mathcal{W}_2^2(\mu_{(1)}(r), \mu_{(2)}(r))  \diff g( r) \right],
\end{aligned}
\end{equation*}
where $g(\cdot) = \int_0^{\cdot} V_r \diff r$, and $C$ depends on the Lipschitz constants $L_{\tilde{\beta}}$ and $L_{\tilde{\sigma}}$. Note that since $\nu \ge 1$, then $V$ is strictly positive and $g$ increasing. Also, $\mu_{(1)}(r) = \mathbb{P}_{(X^{(1)}_r, V_r)}$ and similarly for $\mu_{(2)}(r)$.\newline

\noindent Similarly as in the previous proofs, let $\tau_{\omega}:= \inf\{t\ge 0 \, \lvert \, g(t) \ge \omega\}$, for any $\omega \ge 0$ and $\tau_0 = 0$.\newline 

\noindent Set $\tau = \tau_{\omega}$ for fixed $\omega >0$. Also, consider a stochastic time change $s= g(r)$ so that $r = \tau_s$. Applying the above and using the definition of the Wasserstein metric leads to

\begin{equation*}
\begin{aligned}
\mathbb{E}\left[\underset{t \in [0,T]}{\sup}|X^{(1)}_{t \wedge \tau_{\omega}}- X^{(2)}_{t \wedge \tau_{\omega}}|^2\right] &\le C  \mathbb{E}\left[ \int_0^{g(T) \wedge \omega}\underset{t \in [0,T]}{\sup} |X^{(1)}_{t \wedge \tau_s}- X^{(2)}_{t \wedge \tau_s}|^2 + \mathbb{E}\left[|X^{(1)}_{\tau_s}- X^{(2)}_{\tau_s}|^2 \right] \diff s \right] \\
&\le C \int_0^{\omega} \mathbb{E}\left[\underset{t \in [0,T]}{\sup} |X^{(1)}_{t \wedge \tau_s}- X^{(2)}_{t \wedge \tau_s}|^2 \right] \diff s.
\end{aligned}
\end{equation*}
Applying Gr\"{o}nwall's lemma leads to 
\(\mathbb{E}\left[\underset{t \in [0,T]}{\sup}|X^{(1)}_{t \wedge \tau_{\omega}}- X^{(2)}_{t \wedge \tau_{\omega}}|^2\right] =0. \)
Therefore, for fixed $\omega > 0$ and $t \in [0,T]$, we have $X^{(1)}_{t \wedge \tau_{\omega}} = X^{(2)}_{t \wedge \tau_{\omega}}.$
To finish the proof, we can write for any $\omega > 0 $,
\begin{equation*}
\begin{aligned}
&\mathbb{E}\left[\underset{t \in [0,T]}{\sup}|X^{(1)}_{t}- X^{(2)}_{t}|^2\right]  \le \mathbb{E}\left[\underset{t \in [0,T]}{\sup}|X^{(1)}_{t}- X^{(2)}_{t}|^2 \mathbbm{1}_{\{\tau_{\omega} > T\}}\right] + \mathbb{E}\left[\underset{t \in [0,T]}{\sup}|X^{(1)}_{t}- X^{(2)}_{t}|^2 \mathbbm{1}_{\{\tau_{\omega} \le T\}}\right] \\
& \le \mathbb{E}\left[\underset{t \in [0,T]}{\sup}|X^{(1)}_{t \wedge \tau_{\omega}}- X^{(2)}_{{t \wedge \tau_{\omega}}}|^2 \right] + \mathbb{E}\left[\underset{t \in [0,T]}{\sup}|X^{(1)}_{t}- X^{(2)}_{t}|^2 \mathbbm{1}_{\{\tau_{\omega} \le T\}}\right].
\end{aligned}
\end{equation*}
We've already shown that the first term above is $0$. Looking at the second term, since $\tau_{\omega} \le T$ and $g$ is increasing, then $\omega \le g(T)$. We therefore have
\begin{equation}\label{omegaeqn}
\begin{aligned}
\mathbb{E}\left[\underset{t \in [0,T]}{\sup}|X^{(1)}_{t}- X^{(2)}_{t}|^2\right]  &\le \mathbb{E}\left[\underset{t \in [0,T]}{\sup}|X^{(1)}_{t}- X^{(2)}_{t}|^2 \cdot \frac{g(T)}{\omega}\right]\\
&\le \frac{1}{\omega} \mathbb{E}\left[\underset{t \in [0,T]}{\sup}|X^{(1)}_{t}|^4 + \underset{t \in [0,T]}{\sup}|X^{(2)}_{t}|^4 + 2g^2(T)\right],
\end{aligned}
\end{equation}
where we used $|X^{(1)}_{t}- X^{(2)}_{t}|^2 g(T) \le (2|X^{(1)}_{t}|^2 + 2|X^{(2)}_{t}|^2)g(T) \le  |X^{(1)}_{t}|^4 + |X^{(2)}_{t}|^4 + 2g^2(T).$\newline

From Lemma \ref{aprioriestimates}, $\mathbb{E}\left[\underset{t \in [0,T]}{\sup}|X^{(1)}_{t}|^4]\right]$ and $\mathbb{E}\left[\underset{t \in [0,T]}{\sup}|X^{(2)}_{t}|^4\right]$  are finite for $T<T^{*}$. For the last term in equation \eqref{omegaeqn} notice that by Cauchy--Schwarz inequality,
\begin{equation*}
\mathbb{E}\left[g^2(T)\right] \le T \mathbb{E}\left[\int_0^T V^2_r \diff r \right]  \le T \int_0^T \mathbb{E}\left[\underset{r \in [0,T]}{\sup}V^2_r  \right] \diff r,
\end{equation*}
which, by Lemma \ref{momentboundvol}, is also finite. Sending $\omega \to \infty$ in equation \eqref{omegaeqn} gives
\begin{equation*}
\mathbb{E}\left[\underset{t \in [0,T]}{\sup}|X^{(1)}_{t}- X^{(2)}_{t}|^2\right]  = 0,
\end{equation*}
concluding that $X^{(1)}_{t}= X^{(2)}_{t}$ for all $t \in [0,T]$, as required.
\end{proof}

To establish the existence of a strong solution to \eqref{independentparticles3}, we partition $[0,T]$ into uniform time intervals and use an interpolated Euler-like sequence in which at each time-step of the partition the measure component is held constant. This technique is inspired by Section 3 in \cite{KloLor}, and Section 3 in \cite{LiMaoSongWuYin}. We will show that this sequence is Cauchy in $L^2(\Omega, \mathcal{F},\mathbb{P})$ and by the completeness of $L^2(\Omega, \mathcal{F},\mathbb{P})$ there exists a limit $X:\Omega \times [0,T] \to \mathbb{R}$, which we show to be a solution to \eqref{independentparticles3}.
\begin{proposition}[Existence]\label{existencetheorem}
Let Assumptions \ref{A3} hold and also assume that $\mathbb{E}[|X_0|^2] < \infty$ and $\nu \ge 1$. There exists a strong solution of \eqref{independentparticles3} on $[0,T]$ for $T < T^*$ for $T^*$ as in \eqref{Tbound}.
\end{proposition}
We include a sketch of the proof in the appendix, since it follows using similar arguments as those in Theorem 2.2 of \cite{LiMaoSongWuYin}.

\section{Numerical experiments}\label{numexp}
We now demonstrate our theoretical findings on the strong convergence of the discretised scheme \eqref{EM} in time and the pathwise propagation of chaos through numerical experiments. The main goal of these experiments is to evaluate the strong convergence of the time-discretisation scheme (as presented in Theorem \ref{hlsvEM}) in cases where the condition on the Feller ratio, $\nu > \nu^{*}$, is not satisfied, which is observed in practical applications.\bigbreak

\noindent Instead of using real market data, we simulate a \enquote{market} implied volatility surface using pre--calibrated Heston parameters, $V_0 = 0.0094$, $\theta = 0.0137$, $\kappa = 1.4124$, $\rho = -0.1194$, and $\xi = 0.2988$, sourced from \cite{CozMarRei2}. These parameters were calibrated to an FX market.
\bigbreak

\noindent Recall that we model the spot price process \( S \) using the risk--neutral dynamics \eqref{HestonLVM}. To calibrate this model to market prices, we begin by setting the leverage function to \( 1 \), effectively reducing the model to a pure Heston process. Since real market data is not used in our analysis, we adjust the parameters from our \enquote{market} volatility surface to \( V_0 = 0.0094 \), \( \theta = 0.01 \), \( \kappa = 1.5 \), \( \rho = - 0.1 \), and \( \xi = 0.3 \), treating these as the \enquote{calibrated} Heston parameters. Subsequently, we calibrate the leverage function \( \sigma(t, S_t) \). Using condition \eqref{Dupireiff} for exact calibration, we obtain the calibrated system \eqref{sde}, which we then approximate using the particle system \eqref{particlessystem} as described in Section \ref{particlemethod}.\bigbreak

\noindent Important for the simulation of the particle system \eqref{particlessystem} is the choice of regularisation parameters $\epsilon$ and $\delta$. The authors in \cite{GuyHen} suggest using a time-dependent bandwidth $\epsilon_{t,N} = k_t N^{-1/5}$, where $N^{-1/5}$ comes from the minimisation of the asymptotic mean-integrated squared error (AMISE) of the Nadaraya--Watson estimator, and $k_t$ follows Silverman's rule of thumb (see \cite{Silv}). In our experiments, we fix $\epsilon = S_0N^{-\frac{1}{5}}$, where $S_0$ is the initial value of the stock price, and $\delta = 0.01$. We refer the reader to \cite{ReiTsi} for tests on the accuracy of the particle method applied to the calibrated Heston-type LSV model for different values of the regularisation parameters $\epsilon$ and $\delta$.\bigbreak 

\noindent We first test the strong convergence in time of the Euler--Maruyama scheme defined in \eqref{EM}. We set $T=1$, $S_0 = 100$, $V_0 = 0.0094$, and only use $N=10^3$ particles to save computational time. In Theorem \ref{hlsvEM}, we proved the strong convergence of the EM scheme provided that $\nu >\nu^* = 2+\sqrt{3}$. Here, we consider different values of $\kappa$ and $\xi$ so that we test both cases when the condition on the Feller ratio is satisfied and violated, and check the rate of convergence. Specifically, we run tests for $\kappa_1 = 1.5$, $\kappa_2 = 6$, $\kappa_3 = 18$, while we keep the rest of the parameters fixed as mentioned above, giving $\nu_1 = 0.33$, $\nu_2 = 1.33$, and $\nu_3 = 4.00$ respectively. We observe strong convergence with rates $0.43$, $0.51$, and $0.51$, as shown in Figures \ref{fig:EM1}, \ref{fig:EM2}, and \ref{fig:EM3}, respectively. In the plots below, we show the true data points and fit a regression line to estimate the rate of convergence for each case. We also provide a $95$\% confidence interval for each estimated gradient. Noticeably, for $\nu_2 = 1.33 < \nu^*$, the rate of convergence agrees with our theoretical findings, even if the condition on the Feller ratio is violated. For $\nu_1 = 0.33 <  \nu^*$ the rate is slower, as expected, while for $\nu_3 = 4.00 > \nu^*$ we observe half-order convergence as suggested by the theory.\bigbreak
\begin{figure}[H]
  \centering
  \begin{minipage}[b]{0.45\textwidth} 
    \centering
    \includegraphics[width=\textwidth]{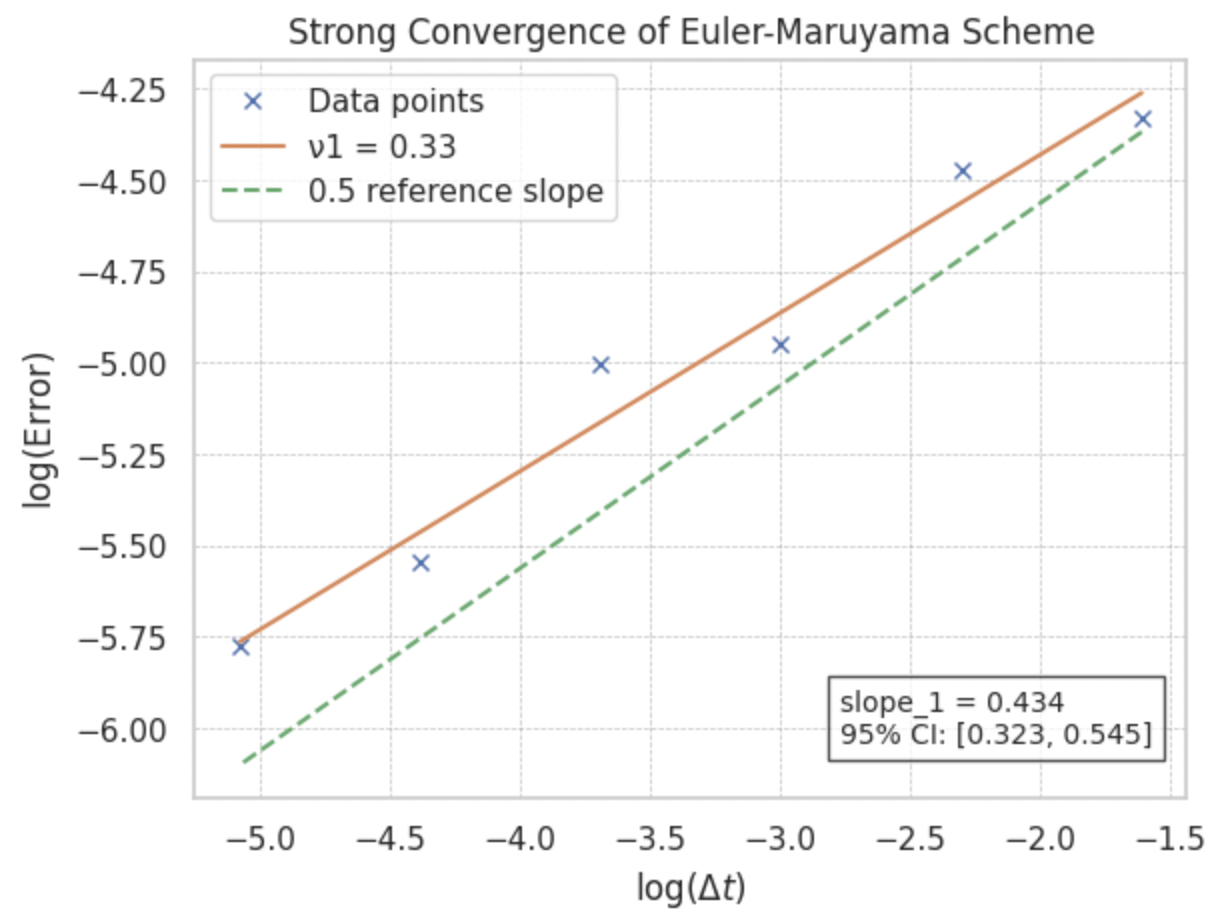}
    \caption{$\kappa_1 =1.5$}
    \label{fig:EM1}
  \end{minipage}
  \hfill
  \begin{minipage}[b]{0.45\textwidth} 
    \centering
    \includegraphics[width=\textwidth]{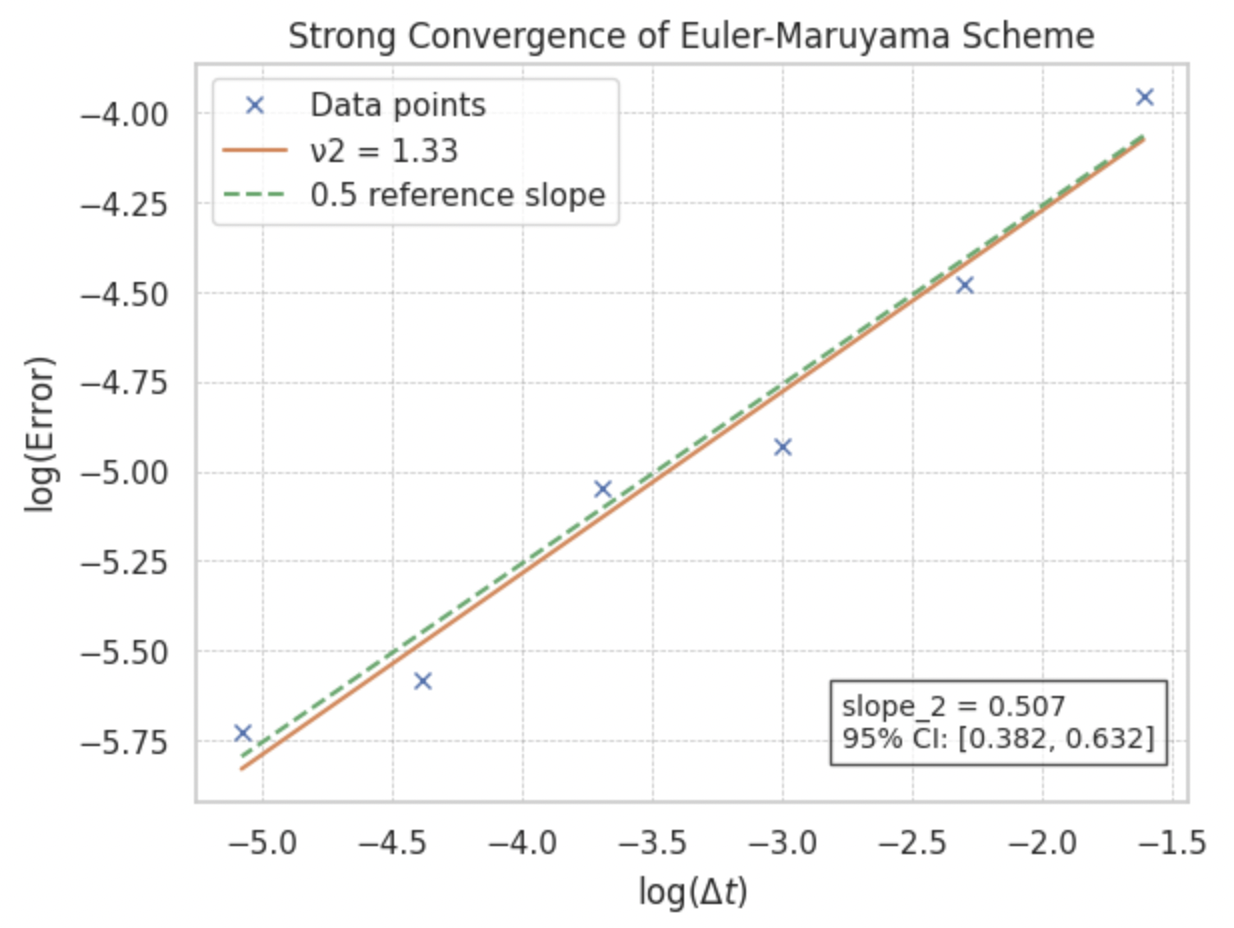}
    \caption{$\kappa_2 =6$}
    \label{fig:EM2}
  \end{minipage}
\end{figure}

\begin{figure}[H]
  \centering
  \begin{minipage}[b]{0.45\textwidth} 
    \centering
    \includegraphics[width=\textwidth]{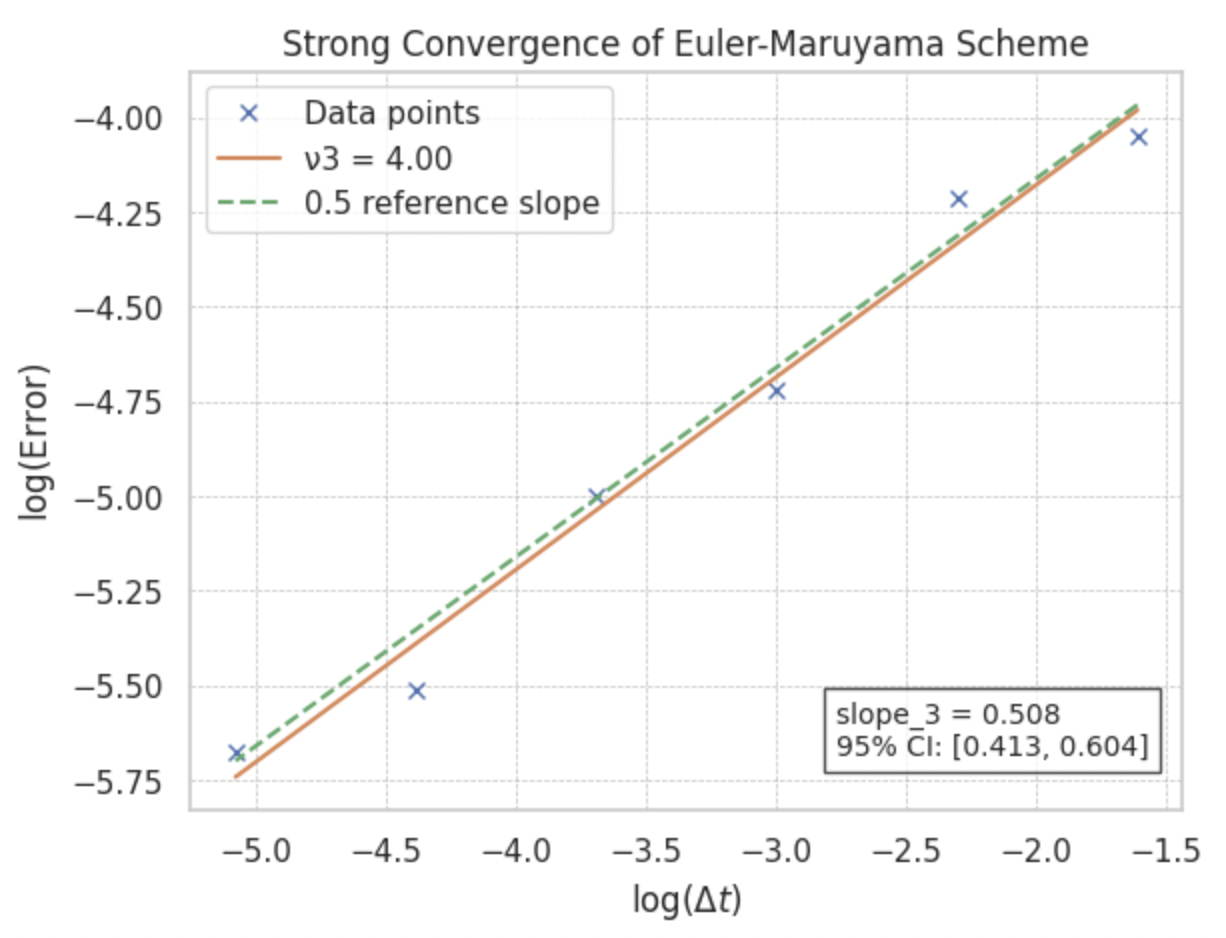}
    \caption{$\kappa_3 = 18$}
    \label{fig:EM3}
  \end{minipage}
  \hfill
  \begin{minipage}[b]{0.45\textwidth} 
    \centering
    \includegraphics[width=\textwidth]{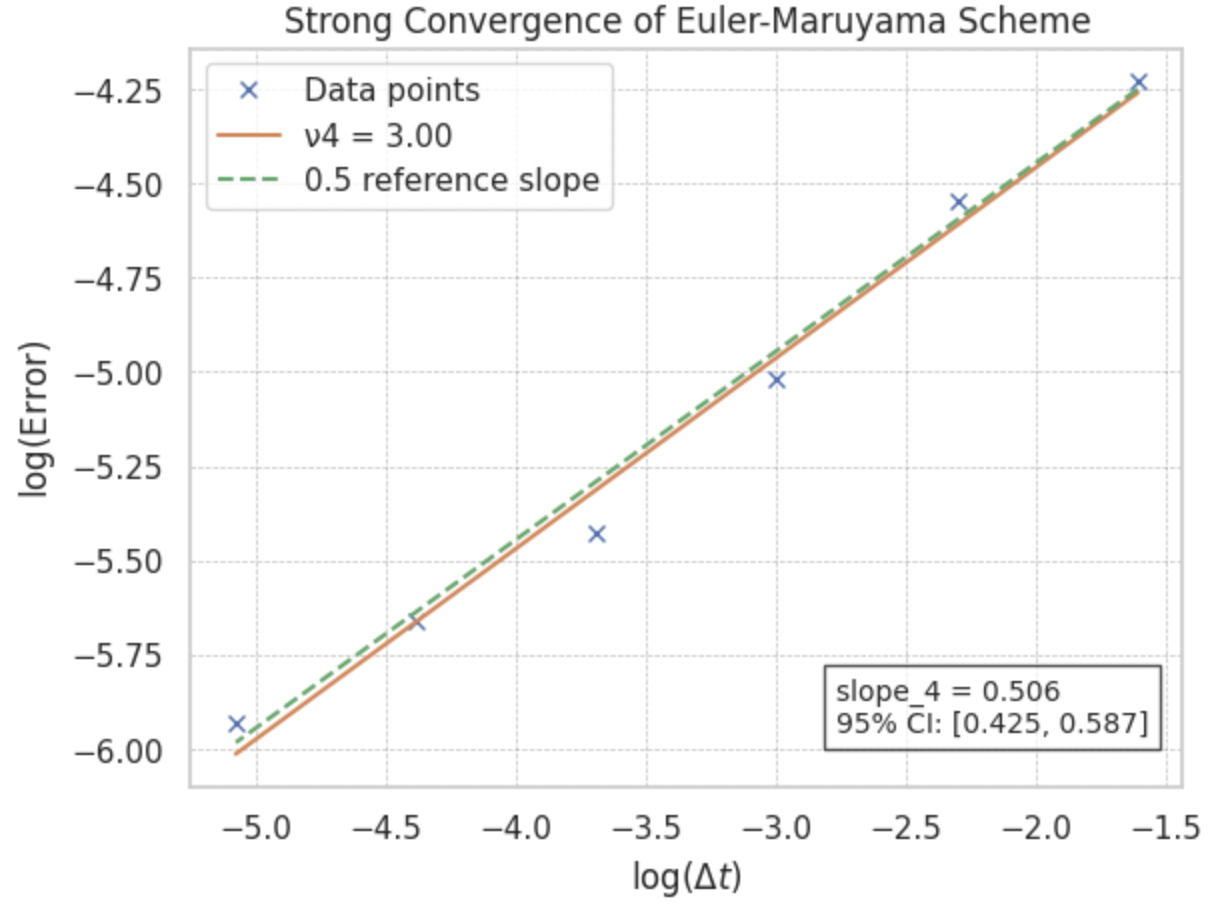}
    \caption{$\xi_4 = 0.1$}
    \label{fig:EM4}
  \end{minipage}
\end{figure}

 \noindent We also run further tests in which we keep $\theta = 0.01$, $\kappa = 1.5$, and $\rho = -0.1$, and let $\xi$ vary. Specifically, we consider $\xi_4 = 0.1$, $\xi_5 = 0.5$, $\xi_6 = 0.8$ giving $\nu_4 = 3.00$, $\nu_5 = 0.12$, and $\nu_6 = 0.05$, respectively. In Figure \ref{fig:EM4}, we observe a rate of $0.51$ for $\nu_4 = 3.00$, which is consistent with the one proved in theory even if $\nu_4 < \nu^*$. For the smaller Feller ratios $\nu_5$ and $\nu_6$ in Figures \ref{fig:StrongConvEM3} and \ref{fig:StrongConvEM4}, the data become noisier. This could occur because in these cases, the value of $\xi$ is larger and therefore the diffusion term in the dynamics of $V$ is larger, adding noise in the process. These instabilities also affect our state process through the leverage function, causing noise in our data.\bigbreak

\begin{figure}[H]
  \centering
  \begin{minipage}[b]{0.45\textwidth} 
    \centering
    \includegraphics[width=\textwidth]{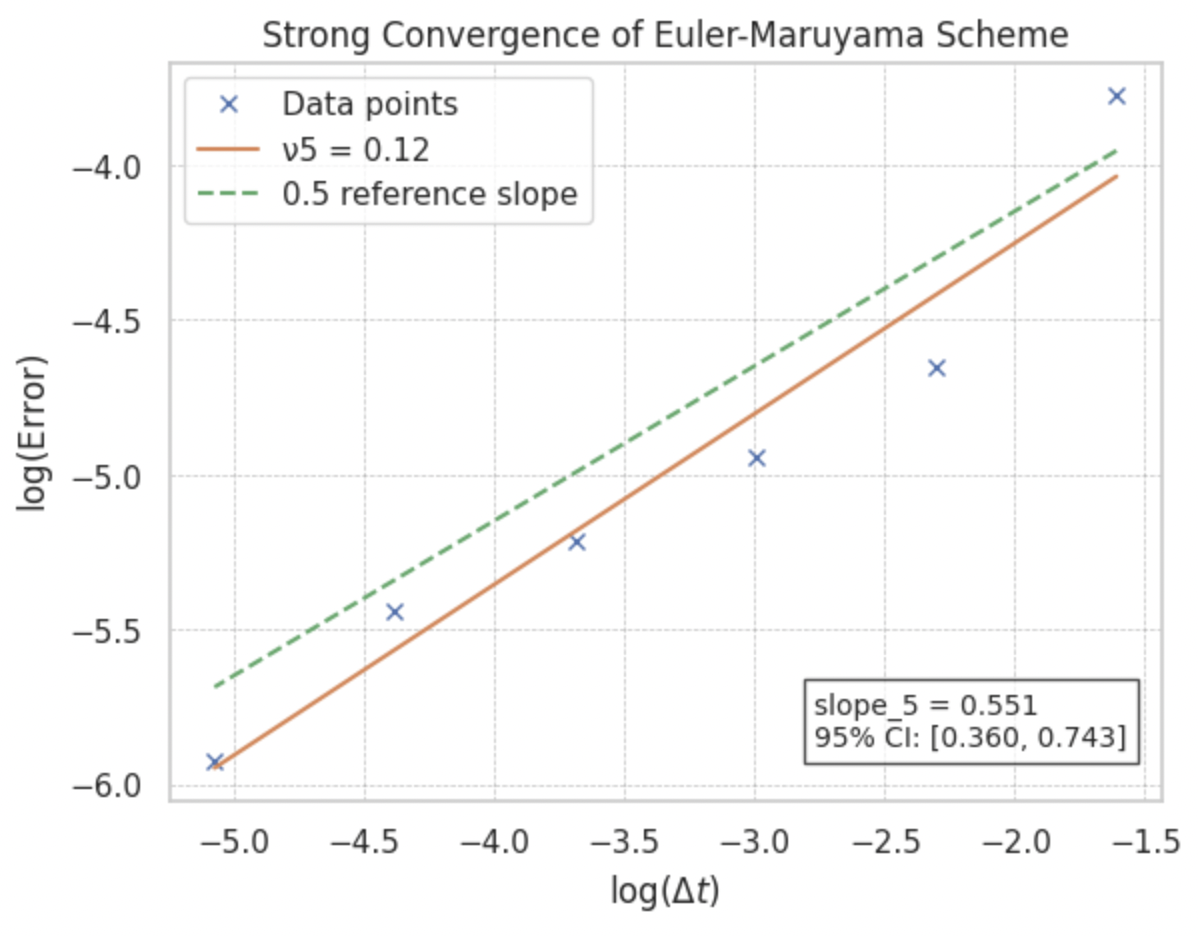}
    \caption{$\xi_5 = 0.5$}
    \label{fig:StrongConvEM3}
  \end{minipage}
  \hfill
  \begin{minipage}[b]{0.45\textwidth} 
    \centering
    \includegraphics[width=\textwidth]{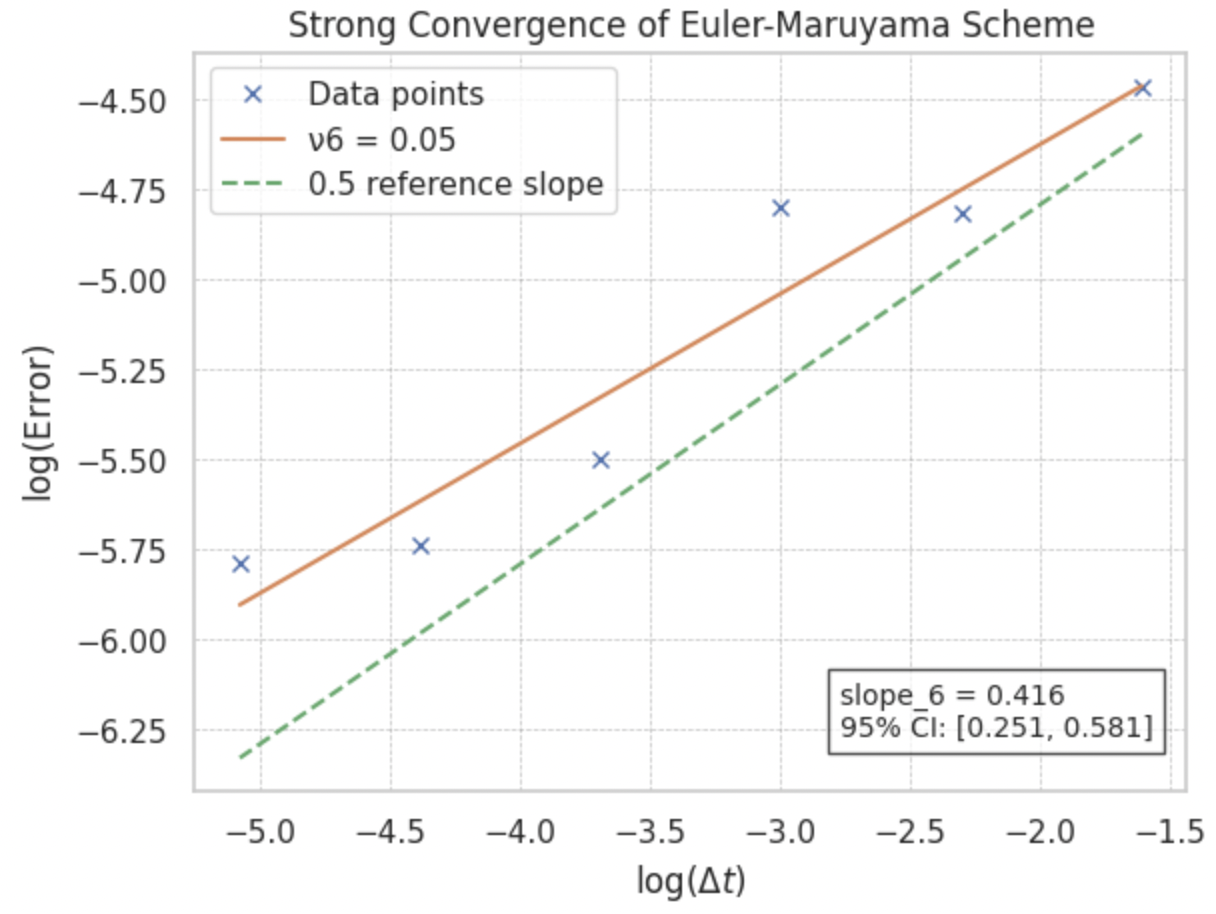}
    \caption{$\xi_6 = 0.8$}
    \label{fig:StrongConvEM4}
  \end{minipage}
\end{figure}

\begin{figure}[H]
  \centering
  \begin{minipage}[b]{0.45\textwidth}
    \centering
    \includegraphics[width=\textwidth]{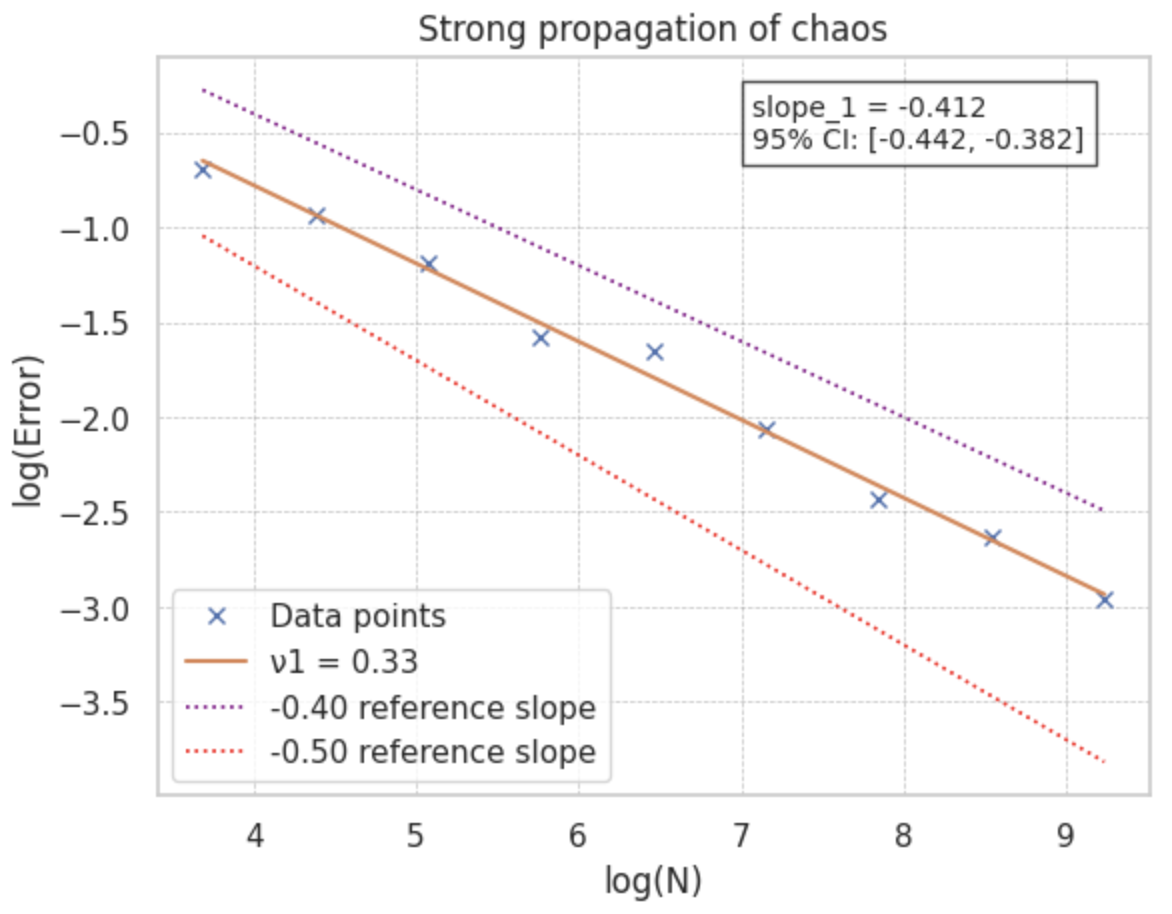}
    \caption{$\kappa_1 = 1.5$}
    \label{fig:prop_chaos_1}
  \end{minipage}
  \hfill
  \begin{minipage}[b]{0.45\textwidth} 
    \centering
    \includegraphics[width=\textwidth]{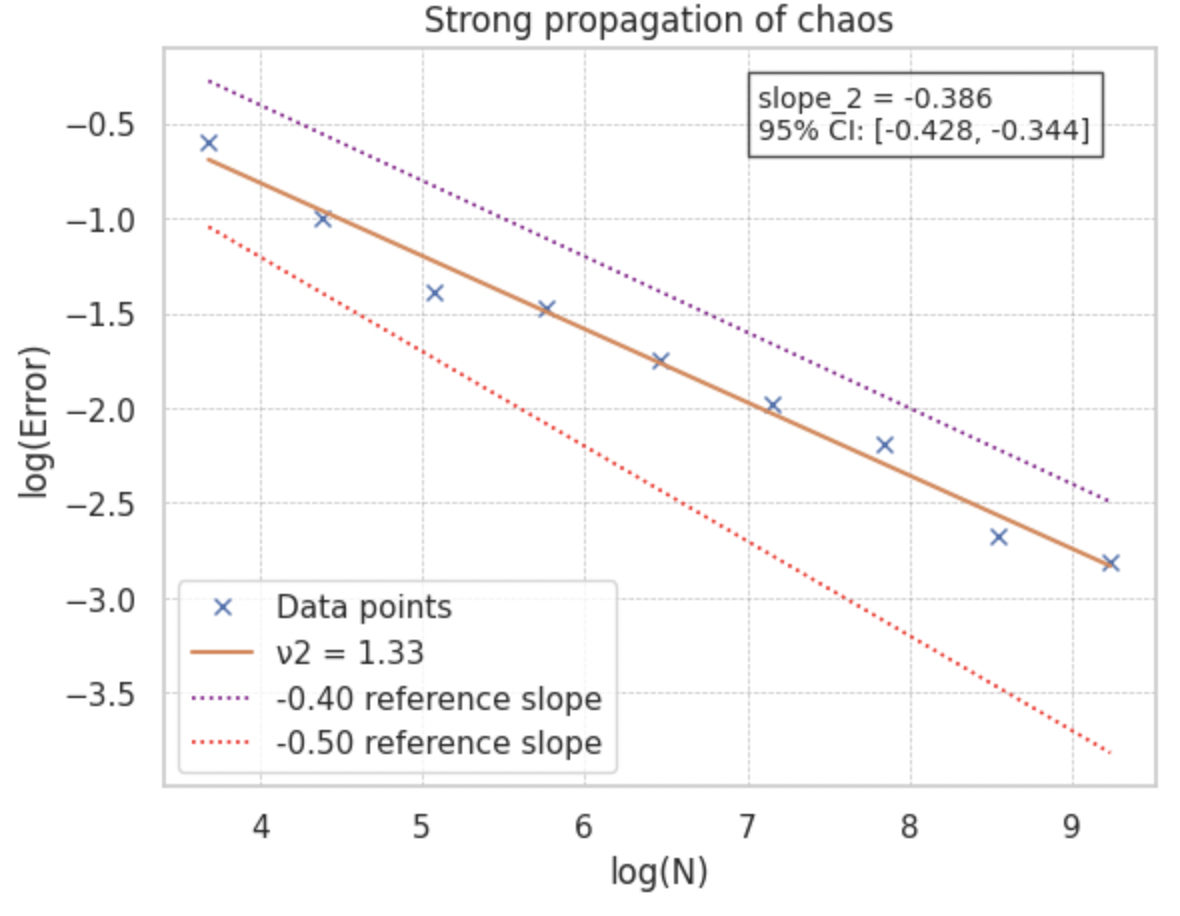}
    \caption{$\kappa_2 = 6$}
    \label{fig:prop_chaos_2}
  \end{minipage}
\end{figure}

\noindent Furthermore, for these different parameters and corresponding Feller ratios, we illustrate the strong propagation of chaos obtained in Theorem \ref{PropChaos}, in which the condition on the Feller ratio is $\nu \ge 1$. Specifically, we plot in log-log scale the root-mean-square error (RMSE) for increasing values of the number of particles $N$. We measure
\begin{equation}
\label{eq:rmse}
\text{error} \;:=\; \sqrt{\frac{1}{2N}\sum_{i=1}^{2N} \bigl(S_{T}^{i,2N} \;-\; \widetilde{S}_{T}^{i,2N}\bigr)^{2}} \,,
\end{equation}
where both systems $\{ S_{T}^{i,2N} \}_{i=1,\ldots,2N}
\quad \text{and} \quad
\{\widetilde{S}_{T}^{i,2N}\}_{i=1,\ldots,2N}$ are driven by the same Brownian motion. However, for the particles denoted by $\widetilde{S}_{T}^{i,2N}$, the leverage function is determined only using the first $N$ particles. In Figures \ref{fig:prop_chaos_1} -- \ref{fig:prop_chaos_6}, we observe a convergence rate of approximately $0.40$, which we compare to the established rate of $0.50$ in the literature in simpler settings (see, e.g., \cite{ReiEngSmi} for the case of super-linear growth in the drift coefficient). We observe that Feller ratios less than $1$ do not significantly impact convergence.
\begin{figure}[H]
  \centering
  \begin{minipage}[b]{0.45\textwidth}
    \centering
\includegraphics[width=\textwidth]{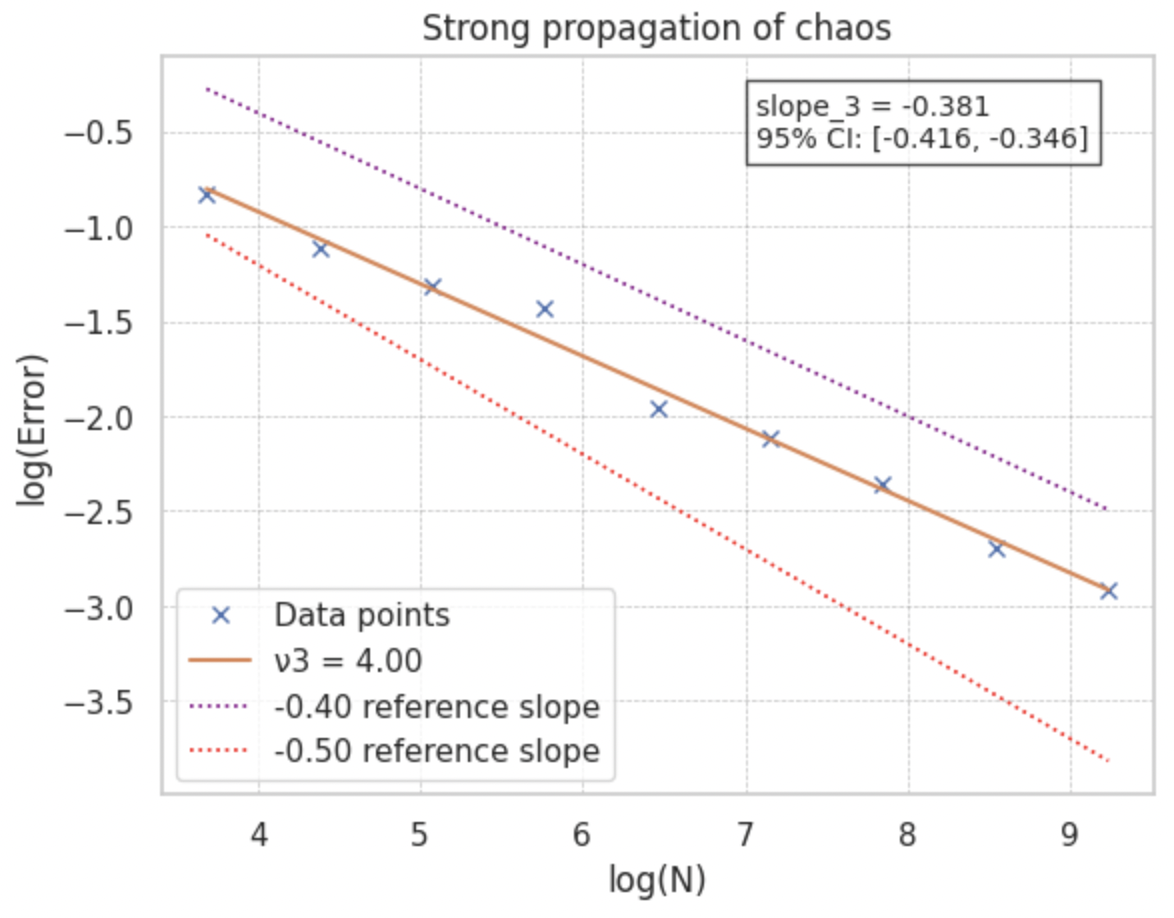}
    \caption{$\kappa_3 = 18$}
    \label{fig:prop_chaos_3}
  \end{minipage}\hfill
  \begin{minipage}[b]{0.45\textwidth}
    \centering
    \includegraphics[width=\textwidth]{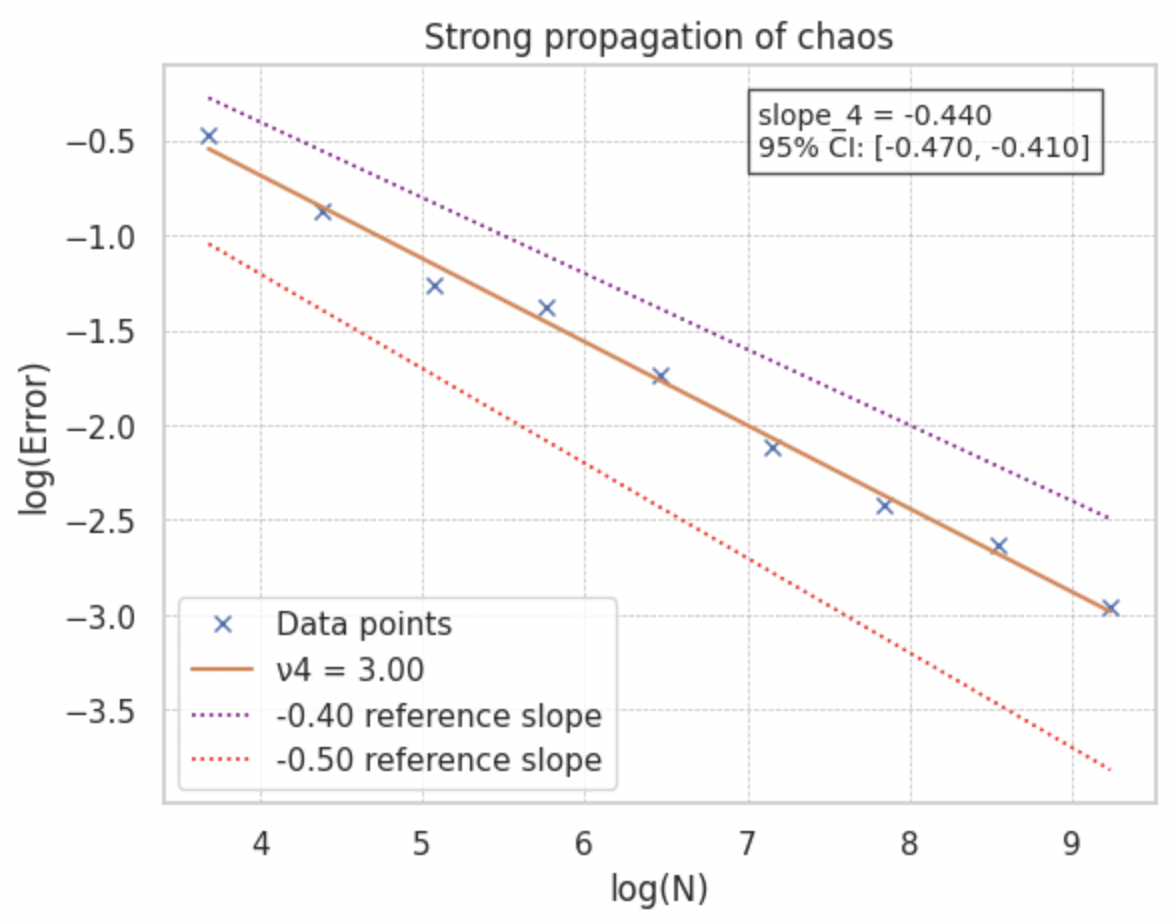}
    \caption{$\xi_4 = 0.1$}
    \label{fig:prop_chaos_4}
  \end{minipage}

  \vspace{1em} 
  \begin{minipage}[b]{0.45\textwidth}
    \centering
    \includegraphics[width=\textwidth]{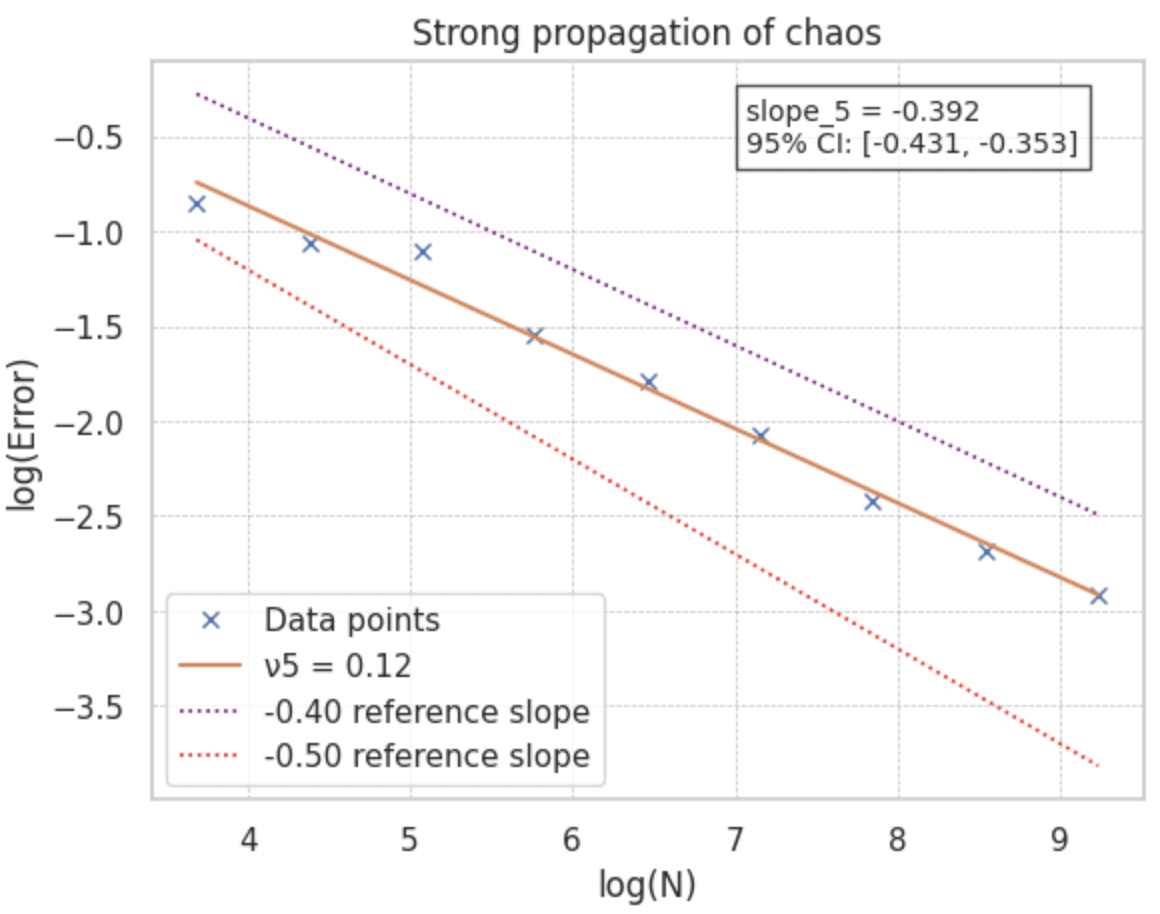}
    \caption{$\xi_5 = 0.5$}
    \label{fig:prop_chaos_5} 
  \end{minipage}\hfill
  \begin{minipage}[b]{0.45\textwidth}
    \centering
    \includegraphics[width=\textwidth]{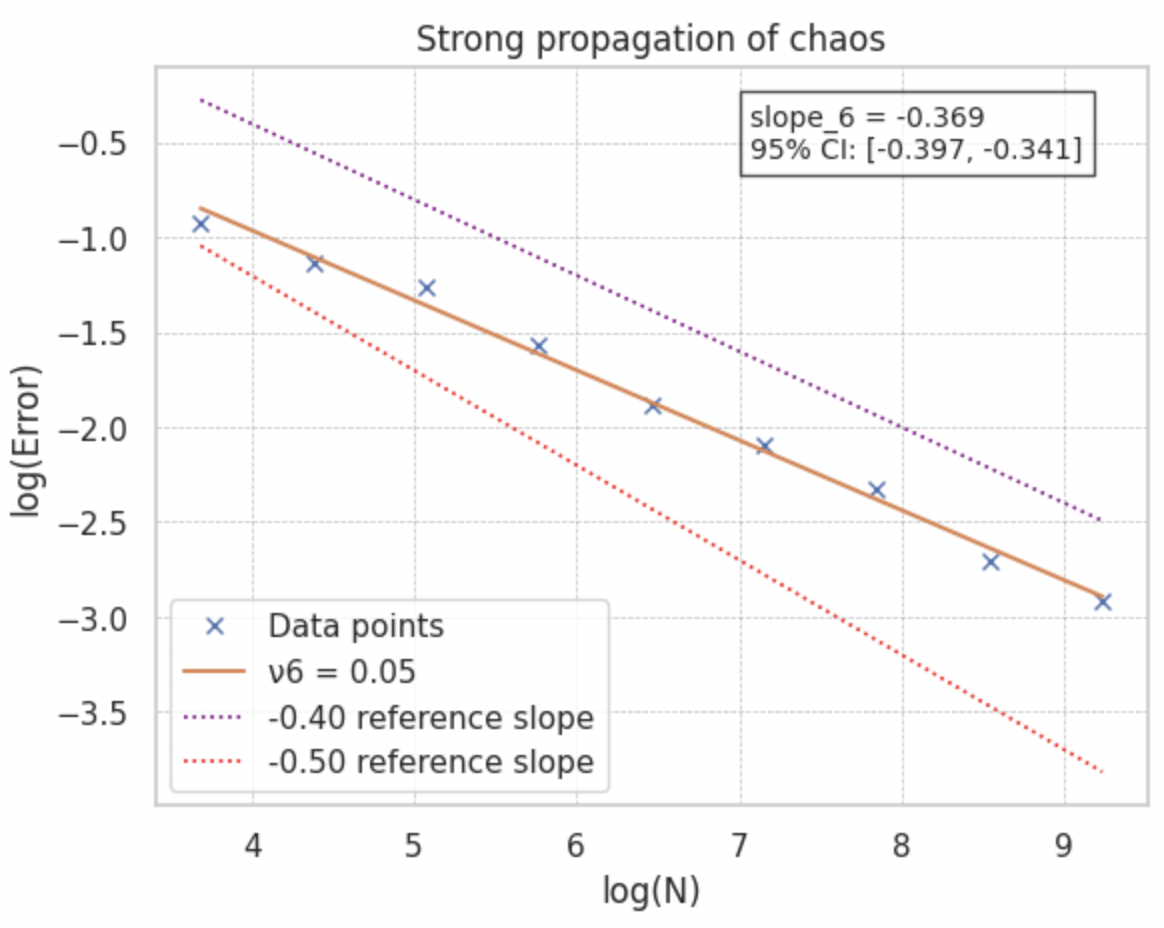}
    \caption{$\xi_6 = 0.8$}
    \label{fig:prop_chaos_6}
  \end{minipage}
\end{figure}
\section*{Acknowledgement}
This work was supported by the Additional Funding Programme for Mathematical Sciences, delivered by EPSRC (EP/V521917/1) and the Heilbronn Institute for Mathematical Research.
\bibliographystyle{plain}


\appendix
\section{Proof of Proposition \ref{existencetheorem}}
\textit{\textbf{Proposition 3.} Let Assumptions \ref{A3} hold and also assume that $\mathbb{E}[|X_0|^2] < \infty$ and $\nu \ge 1$. There exists a strong solution of \eqref{independentparticles3} on $[0,T]$ for $T < T^*$ for $T^*$ as in \eqref{Tbound}.}
\begin{proof}
We can prove the result by following similar arguments as those in \cite{LiMaoSongWuYin}, Theorem 2.2, so here for brevity we summarise the key steps.
\begin{enumerate}[label=\textopenbullet]
\item We start by defining an Euler-like sequence where the measure component is piecewise constant. 

Specifically, let $n \ge 1$ and partition $[0,T]$ to $[t_0 = 0, t_1], (t_1,t_2], ..., (t_{n-1}, t_n = T]$, where for all $m \in \{0,1, ..., n-1\}$, $t_m = \frac{mT}{n} $ and for $t \in (t_{m}, t_{m+1}]$ we define $X^{n}_t$ as

\begin{equation*}
\diff X^{n}_t =  \beta(t,(X^{n}_t,V_t),\mu^{Z^n}_{t_m})\diff t + \sigma(t,(X^{n}_t,V_t),\mu^{Z^n}_{t_m})\diff  W^{x}_t,\,\,\, X^{n}_0=X_0.
\end{equation*}
This is a local SDE since the non--local dependence on the distribution is replaced by $\mu^{Z^n}_{t_m} = \mathbb{P}_{(X^{n}_{t_m},V_{t_m})}$ which is deterministic and independent of $X^{n}_t$ at each time-step. The existence of the above SDE is established (see, for example, Section 2 in \cite{HuangWang}).\bigbreak

\noindent Note that $X^{n} \in L^2(\Omega, \mathcal{F}, \mathbb{P})$ since for $T < T^*$
\[ \mathbb{E}\left[\underset{t\in [0,T]}{\sup}|X^{n}_t|^2\right] \le \sum_{m=0}^{n-1}\mathbb{E}\left[ \underset{t \in [t_m, t_{m+1}]}{\sup}|X^{n}_t|^2\right] < \infty, \]
which follows by similar arguments as those in Lemma \ref{aprioriestimates}. 
\item We can then show that $(X^{n})_{n \in \mathbb{N}}$ is Cauchy in $L^2(\Omega, \mathcal{F},\mathbb{P})$ by showing that 


\begin{equation*}
\underset{n,l \to \infty}{\lim} \sqrt{\mathbb{E}\left[\underset{t\in [0,T]}{\sup}|X^{n}_t - X^{l}_t |^2 \right]} =0.
\end{equation*}
By the completeness of $L^2(\Omega, \mathcal{F},\mathbb{P})$ we conclude that there exists a unique limit $X \in L^2(\Omega, \mathcal{F},\mathbb{P})$ of the sequence $\{X^{n}\}_{n \in \mathbb{N}}$, i.e.,
\begin{equation}\label{Cauchylimit}
\underset{n \to \infty}{\lim}\mathbb{E}\left[\underset{t\in [0,T]}{\sup}|X^{n}_t - X_t |^2 \right] = 0
\end{equation}
\item It remains to show that this limit satisfies equation \eqref{independentparticles3}.
From \eqref{Cauchylimit}, we get
\begin{equation}
\underset{n \to \infty}{\lim}\mathbb{E}\left[|X^{n}_t - X_t |^2 \right] = 0.
\end{equation}
The convergence in mean-square for fixed $t \in [0,T]$ implies that there exists a subsequence $\{X^{n_k}\}_{k=1}^{\infty}$ that converges to $X$ almost surely for all $t \in [0,T]$, i.e.,\newline 
for almost all $\omega \in \Omega$,
\[X^{n_k}_t(\omega) \xrightarrow{} X_t(\omega) \text{ as } k \to \infty.\]
Furthermore, we need to ensure that the measure $\mu^{Z^{n_k}}_{t_m}$ converges to $\mu^{Z}_{t}$ in an appropriate metric. Here, we consider the Wasserstein $\mathcal{W}_2$ distance and show that
\begin{equation*}
\begin{split}
\underset{k \to \infty}{\lim}\underset{t \in [0, T]}{\sup}\mathcal{W}_2(\mu^{Z^{n_k}}_{t_m},\mu^{Z}_{t})^2 &\le 2 \underset{k \to \infty}{\lim}\underset{t \in [0, T]}{\sup}\mathbb{E}\left[|X^{{n_k}}_t - X^{{n_k}}_{t_m} |^2 \right] + 2\underset{k \to \infty}{\lim}\underset{t \in [0, T]}{\sup}\mathbb{E}\left[|X^{{n_k}}_t - X_{t} |^2 \right]\\
& \le 2 \underset{k \to \infty}{\lim} C\frac{T}{{n_k}} = 0.
\end{split}
\end{equation*}
Therefore, since for all $t \in [0,T]$, $\beta(t,.,.,.)$ and $\sigma(t,.,.,.)$ are jointly continuous in $\mathbb{R} \times \mathbb{R} \times \mathcal{P}_2(\mathbb{R} \times \mathbb{R})$, we deduce that for all $t \in [0,T]$ and almost all $\omega \in \Omega$,
$$\underset{k \to \infty}{\lim}\beta(t,(X^{{n_k}}_t,V_t),\mu^{Z^{n_k}}_{t_m}) = \beta(t,(X_t,V_t),\mu^{Z}_{t}),$$ and 
\[\underset{k \to \infty}{\lim}\sigma(t,(X^{{n_k}}_t,V_t),\mu^{Z^{n_k}}_{t_m}) = \sigma(t,(X_t,V_t),\mu^{Z}_{t}).\]
\item We now show that conditions \textit{(iii)} and \textit{(iv)} of Definition \ref{strongsolndefn} hold. Notice that for all $k \ge 1$, $\mathbb{E}\left[\left| \beta(t,(X^{{n_k}}_t,V_t),\mu^{Z^{n_k}}_{t_m}) -\beta(t,(X_t,V_t),\mu^{Z}_{t}) \right|\right] \le 2\tilde{\beta}_{\max}\mathbb{E}\left[|V_t| \right]$. By Lemma \ref{momentboundvol}, $\underset{t\in [0,T]}{\mathbb{E}\left[|V_t| \right]} <\infty$ which concludes that $\{\mathbb{E}\left[\left| \beta(t,(X^{{n_k}}_t,V_t),\mu^{Z^{n_k}}_{t_m}) -\beta(t,(X_t,V_t),\mu^{Z}_{t}) \right|\right]\}_{k\ge 1}$ is uniformly integrable.\newline Similarly, one can show the same for $\{\mathbb{E}\left[\left| \sigma(t,(X^{{n_k}}_t,V_t),\mu^{Z^{n_k}}_{t_m}) -\sigma(t,(X_t,V_t),\mu^{Z}_{t}) \right|^2\right]\}_{k\ge 1}$.\newline

Also, $\left| \beta(t,(X^{{n_k}}_t,V_t),\mu^{Z^{n_k}}_{t_m})-\beta(t,(X_t,V_t),\mu^{Z^{n}}_{t}) \right| \le 2 \tilde{\beta}_{\max}\left|V_t \right|$ and $\underset{t \in [0,T]}{\sup}\mathbb{E}\left[\left|V_t \right| \right] < \infty$, concluding that $\{|\beta(t,(X^{{n_k}}_t,V_t),\mu^{Z^{n_k}}_{t_m})- \beta(t,(X_t,V_t),\mu^{Z^{n}}_{t})|\}_{k \ge 1}$ is uniformly integrable. The same holds for $\{|\sigma(t,(X^{{n_k}}_t,V_t),\mu^{Z^{n_k}}_{t_m})- \sigma(t,(X_t,V_t),\mu^{Z^{n}}_{t})|^2\}_{k \ge 1}$.\bigbreak

\noindent We can now show that for $t \in [0,T]$,
$\int_0^t\beta(s,(X^{{n_k}}_s,V_s),\mu^{Z^{n_k}}_{s_m})\diff s$ converges to\newline $\int_0^t\beta(s,(X_s,V_s),\mu^{Z^{n}}_{s})\diff s$ in the $L^1$ sense. Applying the dominated convergence theorem twice, see, for example, Theorem $4$ in \cite{ANShi}, we can show
\begin{equation}
\begin{split}
& \underset{k \to \infty}{\lim}\mathbb{E}\left[\left|\int_0^t\beta(s,(X^{{n_k}}_s,V_s),\mu^{Z^{n_k}}_{s_m})\diff s - \int_0^t\beta(s,(X_s,V_s),\mu^{Z^{n}}_{s})\diff s \right| \right]\\
& \le \int_0^t  \underset{k \to \infty}{\lim}\mathbb{E}\left[\left|\beta(s,(X^{{n_k}}_s,V_s),\mu^{Z^{n_k}}_{s_m})\diff s - \beta(s,(X_s,V_s),\mu^{Z^{n}}_{s})\right| \right]\diff s\\
& = \int_0^t \mathbb{E}\left[ \underset{k \to \infty}{\lim}\left|\beta(s,(X^{{n_k}}_s,V_s),\mu^{Z^{n_k}}_{s_m})\diff s - \beta(s,(X_s,V_s),\mu^{Z^{n}}_{s})\right| \right]\diff s = 0.
\end{split}
\end{equation}
For the sequence of stochastic integrals $\{\int_0^t\sigma(s,(X^{{n_k}}_s,V_s),\mu^{Z^{n_k}}_{s_m})\diff W_s \}_{k \ge 1}$, we prove convergence in an $L_2$ sense using the BDG inequality. Specifically,
\begin{equation}
\begin{split}
& \underset{k \to \infty}{\lim}\mathbb{E}\left[\left|\int_0^t\sigma(s,(X^{{n_k}}_s,V_s),\mu^{Z^{n_k}}_{s_m})-\sigma(s,(X_s,V_s),\mu^{Z^{n}}_{s})\diff W_s\right|^2 \right]\\
& \le \underset{k \to \infty}{\lim}\mathbb{E}\left[\left(\underset{t \in [0,T]}{\sup}\left|\int_0^t\sigma(s,(X^{{n_k}}_s,V_s),\mu^{Z^{n_k}}_{s_m})-\sigma(s,(X_s,V_s),\mu^{Z^{n}}_{s})\diff W_s\right| \right)^2 \right]\\
& \le C\underset{k \to \infty}{\lim} \mathbb{E}\left[\int_0^T \left|\sigma(s,(X^{{n_k}}_s,V_s),\mu^{Z^{n_k}}_{s_m}) - \sigma(s,(X_s,V_s),\mu^{Z^{n}}_{s}) \right|^2 \diff s\right]\\
& = C \int_0^T \mathbb{E} \left[ \underset{k \to \infty}{\lim} \left|\sigma(s,(X^{{n_k}}_s,V_s),\mu^{Z^{n_k}}_{s_m}) - \sigma(s,(X_s,V_s),\mu^{Z^{n}}_{s}) \right|^2 \right] \diff s =0.\\
\end{split}
\end{equation}

The convergence in mean-square, concludes that we can extract a subsequence that converges almost surely.  
The conclusion follows.
\end{enumerate}
\end{proof}

\end{document}